\documentclass{article}
\usepackage[utf8]{inputenc}

\usepackage[bookmarksnumbered, colorlinks,plainpages]{hyperref}
\hypersetup{colorlinks,linkcolor=blue,citecolor=blue,urlcolor=black}
\usepackage{authblk}
\usepackage{epsf,epsfig}

\usepackage{amssymb}
\usepackage{stmaryrd}
\usepackage{amsmath}
\usepackage{MnSymbol}
\usepackage{float}
\usepackage[caption=false]{subfig}
\usepackage{color,soul}
\usepackage{xcolor}
\usepackage{pgf,tikz}
\usepackage{graphicx}
\usepackage{dcolumn}
\usepackage{bm,bbm}
\usepackage{color,soul}
\usepackage{xcolor}
\usepackage{import}
\usepackage{csquotes}
\usepackage{amsthm}
\usepackage[toc,page]{appendix}
\graphicspath{{./fig/}}
\usepackage{enumerate}
\usepackage{listings}
\usepackage{soul}
\usepackage[]{xcolor}
\usepackage{chngcntr}
\usepackage{cite}

\newtheorem{theorem}{Theorem}
\newtheorem{lemma}{Lemma}
\newtheorem{proposition}{Proposition}
\newtheorem{corollary}{Corollary}
\newtheorem{remark}{Remark}
\newtheorem{example}{Example}
\newcommand{\eref}[1]{Eq.~(\ref{#1})}

\newcommand{\Tr}{{\mathrm{Tr}}}

\newcommand{\1}{{\rm 1\hspace{-0.9mm}l}}

\DeclareMathOperator{\diag}{diag}
\newtheorem{definition}{Definition}

\newcommand{\ket}[1]{|#1\rangle}
\newcommand{\bra}[1]{\langle #1|}

\newcommand{\project}[1]{\ket{#1}\bra{#1}}

\renewcommand{\c}[1]{\mathcal{#1}}

\newcommand{\T}{{\mathrm T}}

\newcommand{\vect}[1]{\mathbf{#1}}

\def\a{\alpha}
\def\b{\beta}

\def\m{\mu}
\def\n{\nu}
\def\be{\begin{equation}}
\def\ee{\end{equation}}
\def\ba{\begin{eqnarray}}
\def\ea{\end{eqnarray}}
\def\la{\langle}
\def\ra{\rangle}
\def\a{\alpha}
\def\nn{\nonumber}
\def\b{\beta}

\def\m{\mu}
\def\n{\nu}

\def\h{\hskip 1cm}

\def\lo{\longrightarrow}
\begin{document}

\title{Quasi-inversion of quantum and classical channels \\
     in finite dimensions}

\author[1,2]{Fereshte Shahbeigi}
\author[1]{Koorosh Sadri}
\author[1]{Morteza Moradi}
\author[2,3]{ Karol {\.Z}yczkowski}
\author[1]{Vahid Karimipour}
\affil[1]{Department of Physics, Sharif University of Technology, Tehran, Iran}
\affil[2]{Center for Theoretical Physics, Polish Academy of Sciences, 02-668 Warszawa, Poland}
\affil[3]{Institute of Theoretical Physics,
         Uniwersytet Jagiello\'{n}ski, 30-348 Krak{\'o}w, Poland}

\date{June 14, 2021}

\maketitle

\begin{abstract}
\noindent
We introduce the concept of quasi-inverse of quantum and classical channels,
prove general properties of these inverses and determine them for
a large class of channels acting in an arbitrary finite dimension.
Therefore we extend the previous results of \cite{kbf} to
arbitrary dimensional channels and to the classical domain.
We demonstrate how application of the proposed scheme can increase
on the average the fidelity between a given random pure state and its image
transformed by the quantum channel followed by its quasi-inversion.
\end{abstract}
\section{Introduction}

It is generally understood that quantum resources
make significant improvements over the classical ones
in most of the  information processing tasks \cite{NC00}.
However, these resources are usually  fragile under  the noise
caused by inevitable interactions with environment which may drastically
neutralize the quantum advantage mentioned  above.
Stated in other words, an open quantum system is usually
interacting with an environment, so
its dynamics cannot be described by a unitary evolution,
$\rho\lo U\rho U^\dagger$.
This unitary dynamics is an idealization which almost never occurs in reality.
There are always inevitable and unknown couplings with the environment
which destroy the coherence, decrease the purity of the state,
 and deteriorate information encoded into a quantum system \cite{GKS76,Li76}.
One of the central results in the quantum theory is that
a general non-unitary dynamics of an open quantum system
can be characterized by operators acting entirely within the quantum system \cite{cho,kra}.
The latter dynamics has long been studied and by now there is an extensive
literature on the subject. The simplest way to describe a non-unitary
 dynamics is to use the Kraus form of a channel acting on a density matrix $\rho$
 of order $d$,
\be\label{channel}
\rho\lo {\cal E}(\rho)=\sum_{\a}K_\a \rho K_\a^\dagger,
\ee
 which can be interpreted as a generalization of the unitary dynamics.
The standard unitarity condition $U^\dagger U=I$ has been replaced above
by the identity resolution, $\sum_{\a=1}^r K_\a^\dagger K_\a=I$.
Here $K_\a$  denotes a Kraus operator of size $d$.
The number of these operators, $r$, may vary, but it is always 
possible to find representations with $r \leq d^2$.
Any map of this class is Completely Positive and
Trace-preserving (CTP) and is called a quantum operation
 or a quantum channel \cite{cho,kra}.
It captures the effect of errors (noise, decoherence and dissipation)
in a quantum system caused by interaction with the  environment. \\

It is easy to see that any unitary evolution can be explicitly inverted by replacing $U$
with $U^{-1}$, so one can get back the original state by turning the
dynamics backward. Even if we set aside practical considerations for reducing
the effect of noise and errors, since after all there exist error correcting codes
and other methods
\cite{lam, lid, bei, kwi,K08,SAZ09,SADZ10,KLKK12,AAV88,POWRW05}
for dealing with these issues,
it is a mathematical curiosity to ask, whether a general quantum channel can  be
inverted too \cite{karol,CAZ21}.
It is however an established fact that a quantum channel can be
exactly inverted only if it is a unitary transformation.\\

By inversion of a given quantum channel we mean here using another
quantum channel, which is physically possible.
One may therefore ask, whether a quantum channel can be quasi-inverted,
in the sense that another quantum channel ${\cal E}^{qi}$ exists,
 such that
${\cal E}^{qi}\circ {\cal E}$ is as close  to the identity map as possible.
This was the approach taken in \cite{kbf}, where the case of qubit channels
was studied in detail. As qubit channels are completely characterized
and classified in \cite{rus,fuj}, the authors of \cite{kbf}
managed to find the quasi-inverse of all qubit channels. It was shown that the
quasi-inverse of every quantum channel, except for a measure zero set,
 is nothing but a suitably defined unitary map \cite{kbf}.
An explicit formula for deriving this unitary map was also derived.\\

Slightly related issue  was earlier considered by Koenig et al. \cite{KRS09}
who analyzed the following problem:
 given a bipartite quantum state $\rho_{AB}$
 one wishes to convert it as closely as possible to a maximally entangled state
 by applying a quantum channel only on the system $B$.
Due to Jamio{\l}kowski isomorphism any quantum channel
acting on a $d$-dimensional system can be treated as a bipartite
state $\rho_{AB}$ of size $d^2$ -- see \cite{karol2}.
Thus the maximal fidelity optimized over all possible channels
can be related to the conditional min-entropy of the state $\rho_{AB}$.
A similar approach was later used by
Chiribella and Ebler \cite{CE16}, who demonstrated
 that an unknown unitary channel
 can be optimally inverted with an average fidelity of $1/d^2$.
\\

 Extension of  results  obtained in \cite{kbf} for single-qubit channels
  for higher dimensions is not  straightforward,
as very little is known concerning the structure of
the set of channels acting in dimensions $d\ge 3$.
In addition to the quartic increase of the number of parameters
with the dimension $d$, which defies
any kind of geometrical picture for these channels,
certain important and simplifying theorems
which hold for qubit channels do not extend to the higher dimensional case.
For instance,  for $d=2$
any unital channel,
which leaves the maximally mixed state invariant, ${\cal E}(I)=I$
belongs to the class of  mixed unitary channels,
 so it can be represented as a mixture of unitary operations \cite{LS}.
 This property does not hold for general $d$ \cite{Tg, Kumm}
 and already for $d\geq3$ there exists the unital channel of
Landau and Streater, which is not mixed unitary \cite{LS},
see the recent study \cite{FK19} for further information on this map.
This is also related to another important difference which concerns
the extreme points of a convex set, i.e. those points which cannot be written
as convex combination of other points of the set.
While the extreme points of the convex set of unital qubit channels are
unitary maps, this is no longer the case for higher dimensional channels.\\

All these properties
make the study of quasi-inverse of quantum channels a non-trivial task.
Nevertheless, we obtain here some general results on higher dimensional
 channels and their quasi-inverses
 and substantiate these results by several examples of
general families of $d-$ dimensional channels.
We provide upper and lower bounds for the performance of a channel after
it is compensated by its quasi-inverse, i.e. ${\cal E}^{qi}\circ {\cal E}$.
In particular, we show that the quasi-inverse of a channel of an arbitrary
dimension $d$
need not  be a unitary map. We also study 
a class of self-quasi-inverse quantum channels including 
the interesting case of Landau-Streater \cite{LS}.\\

In the second part of this work we study an analogous problem
posed for
classical channels, namely for stochastic matrices
which map the set of probability vectors into itself.
This question, left open in \cite{kbf},
is of a special interest in view of the
correspondence between quantum channels and
their classical counterparts \cite{kcpz}.
We show that several results originally formulated for
the quantum domain find
their natural parallels in the classical scenario. \\

The structure of this paper is as follows: In Section \ref{prem},
we collect the preliminary ingredients, in Section \ref{defff}
we define the quasi-inverse channel whose general properties and
examples are respectively explored in Sections \ref{inversegeneral}
and \ref{inverseexamples}.
These studies are extended to the classical domain in Sections
\ref{inverseclassical} and \ref{furtherexamples}.
We conclude the paper with a discussion. Some of the detailed calculations
and proofs are collected in the Appendices.

\section{Preliminaries}
\label{prem}
Let $H_d$ be a complex $d-$dimensional Hilbert space for which we choose
the computational basis $\{|m\ra, m=1,\cdots, d\}$.
Let $L(H_d)$ be
 the space of linear operators on this Hilbert space.
The state of a $d-$dimensional quantum system (a qudit) is described
by a density matrix $\rho$ which is a positive operator of unit trace
acting on this Hilbert space.
The space of all density matrices  is denoted by $D(H_d)$.
This is a convex subset of $L(H_d)$. Any point of this convex set
is described by $d^2-1$ real parameters. Any linear map $A\in L(H_d)$
can be uniquely mapped to a vector $|A\ra\in H_d\otimes H_d$
in the form
\be\label{vectorized}
|A\ra=d(A\otimes I)|\phi^+\ra,
\ee
where $|\phi^+\ra=\frac{1}{\sqrt{d}}\sum_i|i,i\ra$ is
a maximally entangled state. This is called vectorization of a matrix
$A=\sum_{i,j}A_{i,j}|i\ra\la j|$ into a vector $|A\ra=\sum_{i,j}A_{i,j}|i,j\ra$,
with the correspondence between the inner products:

\be
\Tr(A^\dagger B)=\la A|B\ra.
\ee


Taking $d^2-1$  traceless and Hermitian matrices $\Gamma_i$
along with identity $I$ as a complete basis,
one can write any density matrix  as
\begin{equation}
\label{density}
\rho=\frac1d\big(I+\sum_{i=1}^{d^2-1}r_i\Gamma_i\big)=
\frac1d\big(I+{\bf r}\cdot \bm\Gamma\big)
\end{equation}
We normalize $\Gamma_i$ matrices to satisfy
\begin{equation}
\label{gamma}
\Tr\left(\Gamma_i\Gamma_j\right)=d(d-1)\delta_{ij},
\end{equation}
a concrete choice for this basis is given in Appendix \ref{app1}.
We will then have
\be\label{rhosquared}
 \Tr(\rho^2)=\frac{1}{d}\big[1+(d-1){\bf r}\cdot{\bf r}\big].
\ee
The $(d^2-1)$-dimensional vector $\bm{r}=(x_1,x_2,\dots,x_{d^2-1})^\T$
is a real vector called the generalized Bloch vector.
In contrast to the qubit case $d=2$, the convex set of physical states
(positive matrices with unit trace) is no longer a unit ball.
In fact the geometry of this convex set can be quite complicated and only
partial facts are known about low dimensions, i.e. for $d=3$ \cite{goy,karol2}.
The set of pure states $\rho=|\psi\ra\la \psi|$, where
$|\psi\ra=\sum_{i=1}^{d}\psi_i|i\ra$ is a sphere of dimension $2d-2$
which we denote by $S_{2d-2}$. On the other hand,
any pure state has the property $\Tr(\rho^2)=1$ which
in view of Eqs. (\ref{gamma}) and (\ref{rhosquared})  is equivalent to
${\bf r}\cdot {\bf r}=1$. This is a sphere of dimension
$S_{d^2-2}$ which for $d>2$ has a higher dimension than
the set of pure states. Hence the set of pure states is a subset of
this larger sphere.  This larger sphere contains other points which are not 
even states at all,  see  Appendix \ref{app1} for concrete examples.
The necessary and sufficient condition for the Bloch vector to describe a
pure quantum state can be found in \cite{JS01}.
The necessary condition for the Bloch vector to produce
a general mixed state is provided in \cite{K03}.\\

Consider now a quantum channel ${\cal E}$ represented
by its Kraus operators $K_\alpha$ acting on states of dimension $d$ through
Eq.~\eqref{channel}.
The correspondence (\ref{vectorized}) leads to the following representation
 of channels by matrices acting on vectorized form of density matrices,
\be
{\cal E}(\rho)=\sum_\a K_\a \rho K_\a^\dagger \ \ \ \lo\ \ \
|{\cal E}(\rho)\ra=\Phi_{\cal E}|\rho\ra
\ee
where
\begin{equation}\label{superoperator}
\Phi_{\cal E}=\sum_\a K_\alpha\otimes K_\alpha^*,
\end{equation}
is called the superoperator of the map ${\cal E}$ \cite{karol2}
in which $*$ denotes complex conjugation.
Although a quantum channel has many different sets of Kraus operators, connected by
$L_\beta = \sum_{\beta}U_{\a,\b}K_\beta$, where $U$ is a unitary matrix,
it is straightforward to see that $\Phi$ is unique. It also has the property that
$$\Phi_{\cal E\circ \cal E'}=\Phi_{\cal E}\Phi_{ \cal E'}.$$
The quantum channel ${\cal E}$ induces an affine transformation on the generalized Bloch vector ${\bf r}$

\be \label{affine}
{\bf r}\lo {\bf r'}=M{\bf r}+{\bf t}
\ee
where
\be
M_{ij}=\frac{1}{d(d-1)}\Tr(\Gamma_i {\cal E}(\Gamma_j))\ \ \ \ {\rm and}\ \ \ \  t_i=\frac{1}{d(d-1)}\Tr(\Gamma_i {\cal E}(I)).
\ee

 The matrix $M$ of order $d^2-1$ is called  the distortion matrix. Writing the basis $\{I, \Gamma_i\}$ in vectorized form $\{|I\ra,|\Gamma_i\ra\}$ which are vectors of dimension $d^2$, we have the normalization condition
\be
\la I|I\ra=d,\h \la I|\Gamma_i\ra=\la \Gamma_i|I\ra=0,\h \la \Gamma_i|\Gamma_j\ra=d(d-1)\delta_{i,j}.
\ee

The $d^2$ dimensional vector $|\rho\ra$ can be written in terms of the normalized basis vectors
$|\tilde{I}\ra=\frac{1}{\sqrt{d}}|I\ra$ and $|\tilde{\Gamma}_i\ra=\frac{1}{\sqrt{d(d-1)}}|\Gamma_i\ra$, in the symbolic form
\be
|\rho\ra=\frac{1}{\sqrt{d}}\left(\begin{array}{c}1\\ \sqrt{d-1}\ {\bf r}\end{array}\right),
\ee
where the first component is the coefficient of $|\tilde{I}\ra$ and the second component encapsulates the  coefficients of $|\tilde{\Gamma}_i\ra$
as a vector. The quantum channel ${\cal E}$ turns this into the vector
\be
|\rho'\ra=\frac{1}{\sqrt{d}}\left(\begin{array}{c}1\\ \sqrt{d-1}\ M{\bf r}+{\bf t}\end{array}\right).
\ee
This means that in this basis the superoperator can be written as

\be
\label{superoperator-affine}
 \Phi_{\cal E}=\left(\begin{array}{cc}1&0\\ \sqrt{d-1}\
 {\bf t}&M\end{array}\right).
 \ee
 In the case of unital channels the translation vector
 vanishes, ${\bf t}=0$. If the entire matrix $M$ also vanishes
 the corresponding map $\Phi_*$ represents the completely depolarizing channel,
 which sends any state $\rho$  into the maximally mixed state,
 $\c E_*(\rho)= I/d$.\\

The above form of the map $\Phi$,
 which represents the evolution of the Bloch vector ${\bf r}$,
 is also called its {\sl Liouville  representation}  \cite{KSRJO14}.
  On the other hand, it can be interpreted as the
Fano form \cite{Fa83} of the bi-partite state representing the map in the
Choi-Jamio{\l}kowski isomorphism.
Such a state, proportional to the Choi matrix,
\be\label{ChoiMatrix}
C_{\cal E}:=d({\cal E}\otimes I)(|\phi^+\ra\la \phi^+|)=
\sum_{i,j}{\cal E}(|i\ra\la j|)\otimes |i\ra\la j|,
\ee
forms a positive operator with  $\Tr(C_{\cal E})=d$.
It is related to the superoperator of the channel through
the simple relation \cite{karol2}
\be\label{reshuffling}
C_{\cal E}=\Phi_{\cal E}^R,
\ee
where $A^R$ denotes the following reshuffling of the entries of
 a matrix $A$ of order $d^2$
\be\label{reshuff}
A^R_{ij,kl}:=A_{ik,jl}.
\ee

\subsection{Average fidelity of a channel}
The performance of a quantum channel ${\cal E}$ can be studied in
different ways. Among these we choose the input-output fidelity
 averaged over a uniform distribution of pure states,
 $\overline{F}({\cal E})$,
and the entanglement fidelity $F_E({\cal E})$ \cite{SN96}.
These are respectively defined as follows:
\be\label{fid}
\overline{F}({\cal E})=\int_{S_{2d-2}} d\psi\
\la \psi|{\cal E}(|\psi\ra\la \psi|)|\psi\ra,
\ee
where $\int d\psi=1$ and $d\psi=d(U\psi)$ for any unitary, and
\be\label{efid}
F_E({\cal E})=\la \phi^+|{\cal E}\otimes I|(|\phi^+\ra\la \phi^+|)|\phi^+\ra.
\ee
The average fidelity measures how much the input and output states
are similar to each other and the entanglement fidelity measures how much
a maximally entangled state is affected if the channel ${\cal E}$
acts on one part of this state. As we will see the two quantities are related to
each other in a simple way. Below we calculate these quantities
by different methods and each method sheds light on these quantities
 from a different angle. \\

First consider the average fidelity where it can be written as
\ba\label{fidel}
\overline{F}({\cal E})&=&\sum_\a\int_{S_{2d-2}} d\psi \
\la \psi|K_\a|\psi\ra\la \psi|K_\a^\dagger |\psi\ra\cr
&=&\sum_\a\int_{S_{2d-2}} d\psi\ \la \psi|K_\a|\psi\ra
\la \psi^*|K_\a^* |\psi^*\ra
=\Tr(L \Phi_{\cal E}),
\ea
where
\be
L:=\int_{S_{2d-2}}  d\psi\  |\psi\ra\la\psi|\otimes |\psi^*\ra\la\psi^*|,
\ee
is the isotropic state and $\Phi_{\cal E}$ is defined in Eq. (\ref{superoperator}).
Since $L$ is an isotropic operator in the sense that
$[L,U\otimes U^*]=0\ \ \ \forall\  U$,
it can be written in the form \cite{HH99},
\be\label{isostate}
L=\frac{1}{d(d+1)}(I+\sum_{i,j}|i,i\ra\la j,j|).
\ee
Inserting (\ref{isostate}) in (\ref{fidel}) and using (\ref{reshuffling})
and (\ref{reshuff}) one finds the following
 formula for the average fidelity
\ba\label{fbar}
\overline{F}({\cal E})&=&\frac{1}{d(d+1)}\big[d+\Tr\Phi_{\cal E}\big].
\ea
More details on these calculations will be given in Appendix \ref{A1}.
From \eqref{superoperator} and \eqref{superoperator-affine}
we see that this can also be written as
$\overline{F}({\cal E})=\frac{1}{d(d+1)}\left[d+\sum_\a |\Tr K_\a|^2\right]$ and $\overline{F}({\cal E})=\frac{1}{d(d+1)}\left[d+1+\Tr M\right]$.\\

Entanglement fidelity is derived in a simpler way, e.g. by direct expansion
of the state $|\phi^+\ra=\frac{1}{\sqrt{d}}\sum_i |i,i\ra$,
and hence one finds
\be\label{entanglement fidelity}
F_E({\cal E})=\frac{1}{d}\la \phi^+|\sum_\a K_\a |i\ra\la j|K_\a^\dagger
\otimes |i\ra\la j|\phi^+\ra=\frac{1}{d^2}\sum_\a |\Tr K_\a|^2.
\ee
Therefore, we find the simple relation \cite{HHH99}
\be\label{input-output nd entanglement fidelity}
\overline{F}=\frac{1}{d+1}(1+d\ F_E({\cal E})),
\ee
showing that quantities \eqref{fid} and \eqref{efid}
 are monotone with respect to each other.

\section{Definition of the quasi-inverse channel}\label{defff}
In the quantum communication context, it is usually the case that Alice
(the information source), generates a message state $\rho_M$
from a distribution corresponding to her language alphabet
which is known to Bob (the receiver).
Let us denote by $\mu$ this probability measure over $D(H_d)$,
the set of all quantum states.
The state is then fed into a quantum channel $\mathcal{E}$ to reach Bob.
The received state is denoted by $\rho_R$.
Assuming that the characteristics of the channel are well known to Bob,
and before performing any quantum error correction,
he may want to pass the received state through some other channel
in order to process it to a state more $similar$ to the one Alice has sent.
This latter channel clearly cannot depend on $\rho_M$
since Bob does not know it. It only depends on the channel
$\mathcal{E}$ and the probability measure $\mu$.
The second channel is applied to somehow invert the action of
$\mathcal{E}$, therefore, it is natural to call it the quasi-inverse of
$\mathcal{E}$  \cite{kbf} and it is expected that this channel does this
inversion in the best possible way, given the constraint of being a CPT map.
Therefore, we have the formal definition of the quasi-inverse:

\begin{definition}\label{defin}
	The quasi-inverse of a
	channel $\c E$ is defined as a CPT map ${\cal E}^{qi}$ fulfilling
	the following condition \cite{kbf}
	\begin{equation}\label{quasi definition}
	\overline{F}(\c E^{qi}\circ\c E)\geq \overline{F}(\c E'\circ\c E)
	\quad\forall\c E',
	\end{equation}
where $\overline{F}$ is the average of a proper fidelity function between the output state and the input pure state.
\end{definition}

\begin{remark} \label{remark}
As we will see in the sequel, the quasi-inverse of a quantum channel is unique except when the channel falls on a set of measure zero.
In the simplest case of unital qubit channels,
 in which every channel is unitarily equivalent to a Pauli channel, ${\cal E}(\rho)=\sum_{i=0}^3p_i\sigma_i\rho \sigma_i$,
 only the channels with two or more equal weights $p_i$
  have non-unique quasi-inverses.
 \end{remark}

We now prove our first theorem whose validity does not depend on
the specific form of the fidelity function, nor on the input state restricted
to be pure, rather it depends on two very general properties of the
fidelity measure, as described in the proof.

\begin{theorem}\label{boundary point}
	For any proper similarity function $F$, probability measure $\mu$
	and quantum channel $\mathcal{E}$, a quasi-inverse $\mathcal{E}^{qi}$
	can be found on the boundary of the set of allowed channels.
\end{theorem}
\begin{proof}
	
A proper similarity measure $F$ should satisfy the following two conditions:
\textbf{(a)} \be\label{property a} F(\rho, \rho^\prime)\leq F(\rho, \rho)
\hspace{4mm}\forall\rho, \rho^\prime\in D(H_d); \ee
the state most similar to a state $\rho$ is the state $\rho$ itself, and
\textbf{(b)} concavity:
\be\label{property b}
F(\rho, \alpha\rho_1+(1-\alpha)\rho_2)
\geq\alpha F(\rho, \rho_1)+(1-\alpha)F(\rho, \rho_2).
\ee
One such measure is the well-known fidelity,
but the following argument is independent of the particular form of this
measure.  We proceed with a  proof by contradiction.\\

Let us denote the set of all possible quantum channels on $H_d$ by
$\c C_d$. This is a convex set.  Assume the quasi-inverse
for the channel  ${\cal E}\equiv (M,{\bf t})$,
defined by Definition \ref{defin}, is a quantum channel
${\cal E}'$ in the interior of $\mathcal{C}_d$, figure (\ref{boundary1}).
Now take the inverse of the above affine map which is given by
$(M^{-1},-M^{-1}{\bf t})$. This affine map may not correspond to a
legitimate quantum channel and may be outside $\c C_d$.
Denote it by ${\cal E}^{-1}$. However for small enough $\varepsilon$,
the channel $(1-\varepsilon)\c E'+\varepsilon {\cal E}^{-1}$ is a CPT
$\in \c C_d$.
We now note that this channel performs better in quasi-inverting
the channel ${\cal E}$ since using (\ref{property a}) and (\ref{property b}) we find
\begin{eqnarray*}
	&&\int{\rm d}\mu\ F
	\Big(\rho, \big[(1-\varepsilon)\c E'+
	\varepsilon \c E^{-1}\big]\circ {\cal E} (\rho)\Big)\geq\\&&
	 (1-\varepsilon)\int {\rm d}\mu\  F
	\left(\rho, \c E'\circ \c E(\rho)\right)+
	\varepsilon\int {\rm d}\mu\  F(\rho, \rho)=\\&&
	\int {\rm d}\mu\
	F(\rho, \c E'\circ \c E)\rho)+
	\varepsilon(\int {\rm d}\mu\ \big[F(\rho, \rho)-
	F\left(\rho, \c E'\circ \c E(\rho)\right)\big])\geq\\&&
	\int{\rm d}\mu\ F\left(\rho, \c E'\circ \c E(\rho)\right).
\end{eqnarray*}
This means we have found a linear path along which the average
fidelity increases or stays constant as we go toward the boundary of $\c C_d$.
Clearly, the quasi-inverse has to be in the end of such a path
and hence at the boundary of $\c C_d$.\\
		
	\begin{figure}
		$${\includegraphics[scale=0.23]{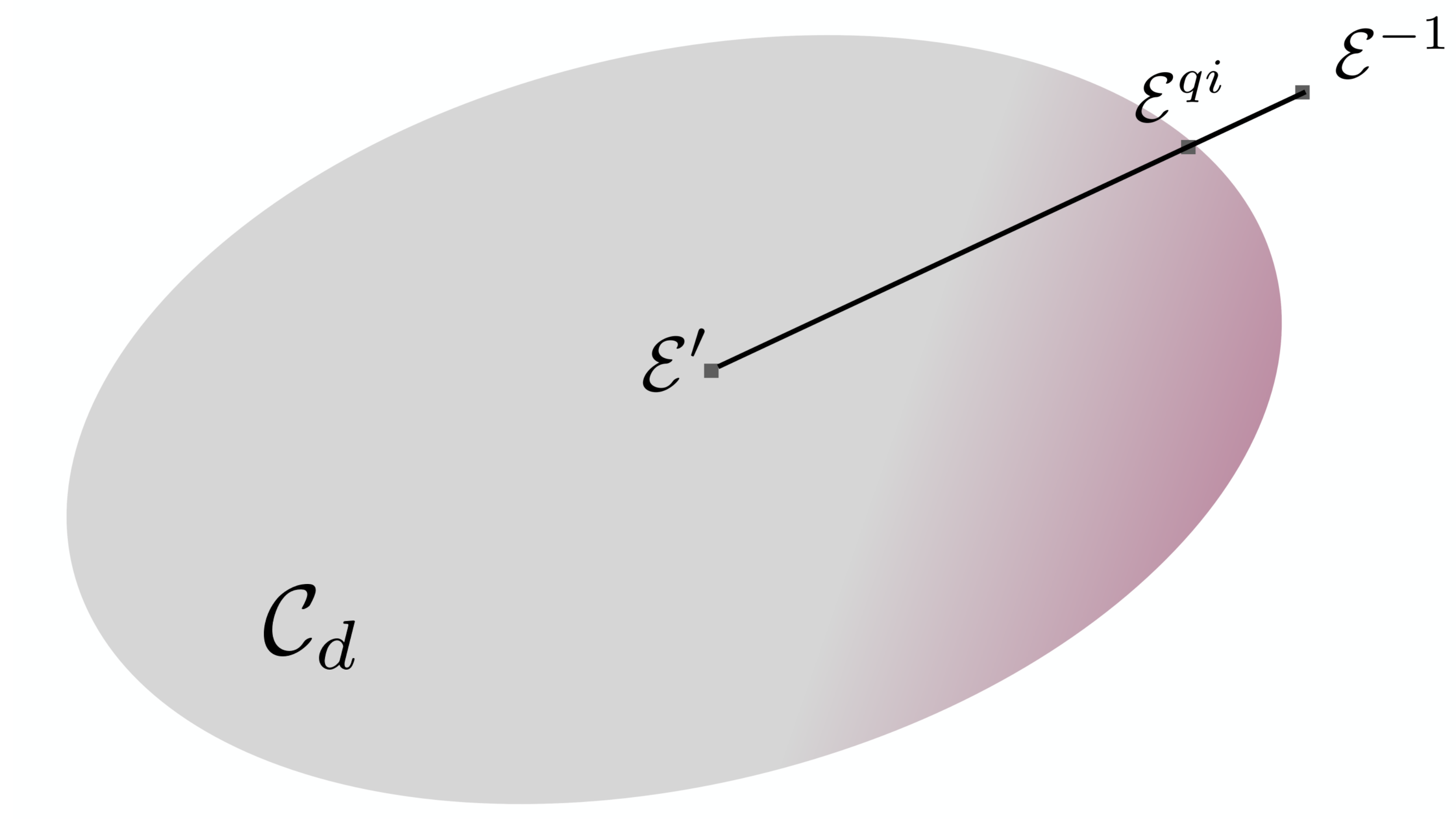}}$$\vspace{-0.4cm}
		\caption{
	Assume that the requested quasi-inverse of a given channel ${\cal E}$ is
	${\cal E}'$, located inside the convex set  $\c C_d$ of all channels.
	Consider the true inverse  map ${\cal E}^{-1}$  (defined in the text),
	which may not belong to $\c C_d$.
	Then at the line joing ${\cal E}'$ and  ${\cal E}^{-1}$
	we always find a legitimate quasi-inverse with a higher average fidelity.
	 Hence the true quasi-inverse ${\cal E}^{qi}$  lies on the boundary
	   of the set  $\c C_d$ of all channels.
	  }
		\label{boundary1}
	\end{figure}

To complete the proof, we need to consider the case of singular channels.
This corresponds to $\det M=0$ and therefore has unit co-dimension.
It means for any singular channel $\mathcal{E}$, we may come up with a
sequence of nonsingular channels $\mathcal{E}_n\rightarrow\mathcal{E}$.
So far we have proved that
	$$\forall\ \  \c E'\in\textrm{int}\ \mathcal{C}_d \hspace{3mm},
	\hspace{3mm}\exists\ \
	\mathcal{E}^{qi}_n\in\partial\mathcal{C}_d \hspace{2mm}$$ such that
	$$\hspace{2mm}\langle F\big(\rho,
	\mathcal{E}^{qi}_n\circ\mathcal{E}_n(\rho)\big)
	\rangle\geq\langle F\big(\rho, \c E'\circ\mathcal{E}_n(\rho)\big)\rangle,$$
where $\langle X\rangle$ denotes the average  of $X$ over $\mu$.
Using the continuity of $F$, implied by its concavity,
the proposition is proved for the singular channels as well.
Note that in this proof we have not assumed any particular form for
the fidelity measure, except the two natural properties (\ref{property a})
and (\ref{property b}), neither we have assumed the average fidelity
to be defined only for pure states.
\end{proof}

Hereafter we assume the standard fidelity measure,
$F(\rho, \rho^\prime)=
\Big(\Tr \sqrt{\sqrt{\rho}\rho^\prime\sqrt{\rho}}\Big)^2$,
of Uhlmann and Jozsa \cite{Uhl76,Jo94},
and use the average
input-output fidelity  (\ref{fid}) for the performance of the channel.
In such a case the average value can be compared
with the mean fidelity between two random quantum states
averaged over the set of all mixed states
with an appropriate measure \cite{ZS05}.\\

Under these assumptions, the quasi-inversion is defined as the channel
maximizing  $\overline{F}(\c E'\circ\c E)$.
In view of Eq. (\ref{fbar}), the practical method for finding
the quasi-inverse is through one of the following maximization problems:
\begin{eqnarray}\label{relation of corrected fidelity}
\overline{F}(\c E^{qi}\circ\mathcal{E})&=&
\nonumber\max_{\Phi'}\frac{1}{d+1}\left(
1+\frac{1}{d}\mathrm{Tr}\Phi'\Phi\right)\\ &=&\nonumber
\max_{\{K'_\beta\}}\frac{1}{d+1}\left(1+\frac{1}{d}\sum_{\alpha,\beta}
|\Tr K'_\beta K_\alpha|^2\right)\\
&=&\nonumber\max_{M'}\frac{1}{d}
\left(1+\frac{1}{d+1}\mathrm{Tr}M'M\right).\\
\end{eqnarray}

These equations immediately imply
that for a given channel, the left and right quasi-inverses are the same, i.e.
$\overline{F}\left(\c E^{qi}\circ\c E\right)=
\overline{F}\left(\c E\circ\c E^{qi}\right)$.
It is an interesting fact with practical benefit. Either Bob can apply the
quasi-inverse after receiving the state or Alice before sending the state to Bob.
Furthermore,  note that for $\c E$ and $\c E'$ denoted by $(M,\bm{t})$ and
$(M',\bm{t}')$, respectively, their concatenation, $\c E'\circ\c E$,
is represented by $(M'M, M'\bm{t}+\bm{t}')$.
This implies that  the translation vector ${\bf t}$, which determines the
non-unitality of the channel, plays no role directly in amount of fidelity and
fidelity after correction. However, it affects the range of the allowed
values of the distortion matrix elements. This is one of the
 features of quasi-inverse for qubit channels \cite{kbf},
which survives for higher dimensions.\\

It is worth mentioning that one may find a relation between the
input-output fidelity after correction and the conditional min-entropy.
The latter quantity for a bipartite state $\rho_{AB}$ is defined as \cite{R05}
\begin{equation}\label{min-ent}
H_{\min}(B|A)_\rho=-\inf_{\sigma_A}\inf_\lambda
\{\lambda\in\mathbb{R}|\ \rho_{AB}\leq2^{\lambda}(\sigma_A\otimes I)\},
\end{equation}
where $\sigma_A$ is a quantum state.
Indeed, it has been proven that \cite{KRS09}
\begin{equation}
2^{-H_{\min}(B|A)_\rho}=d\max_{\c E'}\bra{\phi^+}\ (\c E'\otimes I)
[\rho_{AB}] \ket{\phi^+}.
\end{equation}
Let us assume that $\rho_{AB}$ is the Jamio{\l}kowski state
(the normalised form of the Choi matrix \eqref{ChoiMatrix}) assigned
to a quantum channel  and denoted by $J_\c E=C_{\c E}/d$. Then we get
\begin{eqnarray}\label{op-min-ent}
2^{-H_{\min}(B|A)_{J_\c E}}=d\max_{\c E'}\bra{\phi^+}
(\c E'\circ \c E\otimes I)[\project{\phi^+}]\ \ket{\phi^+}=
d\max_{\c E'} F_E(\c E'\circ\c E)=d F_E(\c E^{qi}\circ\c E).
\end{eqnarray}
Here $F_E(\c E'\circ\c E)$ is the entanglement fidelity \eqref{entanglement fidelity}
of the composed map $\c E'\circ\c E$. So the above equation
shows that entanglement fidelity after correction (thus input-output fidelity
after correction, see \eqref{input-output nd entanglement fidelity}) are
directly related to the conditional min-entropy of  the Choi matrix of the channel.

\section{General properties of the quasi-inverse channel}\label{inversegeneral}
In this section we elaborate on some general properties of the quasi-inverse of
quantum channels and discuss the similarities and the  crucial differences
with the qubit case \cite{kbf}.
Thus far, we have seen in Theorem \ref{boundary point}
that for any proper similarity measure the quasi-inverse
lies on the boundaries of the set of quantum channels.
In what follows, we show that if this similarity measure is linear with respect
to quantum channels, we can specify the quasi-inverse in the set of {\sl extreme}
channels, a subset of the boundary points.
Recall that a point 
of a convex set $\Omega$
is called  extreme  if it cannot be
written as a convex combination of two other points of $\Omega$.
We will use the fact that the set  $\c C_d$
of quantum channels of dimension $d$ is convex \cite{karol2} for any $d$.
\\

\begin{theorem}\label{extreme-ness}
The quasi-inverse of a quantum channel can always be taken to
be an extreme channel.
\end{theorem}

\begin{proof} Assume that the quasi-inverse of
a channel ${\cal E}$ is in the form
\be
{\cal E}^{qi}=\sum_i \lambda_i {\cal E}_i.
\ee
Let ${\cal E}_m$ be the element in the above set for which
$\overline{F}({\cal E}_m\circ {\cal E})\geq \overline{F}({\cal E}_i\circ {\cal E})\ \ \forall\ \ i$.
Then using the linearity of the quasi-inverse, we will have
\be
\overline{F}({\cal E}_m\circ {\cal E})\geq
\overline{F}(\sum_i \lambda_i {\cal E}_i\circ {\cal E})=
\overline{F}({\cal E}^{qi}\circ {\cal E}),
\ee
which means, according to Definition \ref{defin}, that the quasi-inverse can always be taken as an extreme channel.
\end{proof}
The crucial difference between the qubit case and the higher dimensional case
is that for qubit channels quasi-inversion is unital and the extreme points
of the set of one-qubit unital maps are unitary channels, while this is no longer
the case for $d> 2$. Even more than that, not all extreme points of the set of
channels are known for $d>2$, not even for unital channels. \\

Note that  the linearity of $\overline{F}(\c E'\circ \c E)$ over $\c E'$ implies
that if $\c E^{qi}_1$ and $\c E^{qi}_2$ are both quasi-inversions of
$\c E$, then any convex combination of them,
$p\c E^{qi}_1+(1-p)\c E^{qi}_2$, is the inverse as well.
In accordance with the above theorem we arrive at
the following result.
\begin{corollary}\label{uniqueness}
The quasi-inverse
is either unique or an infinite number of them exist where
at least two of them are extreme channels.
\end{corollary}

An immediate result of Theorem \ref{extreme-ness} is that the quasi-inversion
is not an involution, i.e. $(\c E^{qi})^{qi}\neq\c E$ for a general $\c E$.
It is because a quasi-inverse map is an extreme point. So
even if one takes into account the non-uniqueness of quasi-inverse,
see Remark \ref{remark} and Corollary \ref{uniqueness},
there are always non-extremal maps which are not quasi-inverse of any other maps.

\begin{proposition}
Let $\c E^{qi}$ denote the quasi-inverse of $\c E$. Then
 $(\c E^{qi}\circ\c E)^{qi}$ can be taken to be the identity.
\end{proposition}
\begin{proof}
The quasi-inverse $(\c E^{qi}\circ\c E)^{qi}$ is the map which maximizes
the input-output fidelity
\begin{equation}
\max_{\c E^{\prime\prime}}
\overline{F}\left(\c E^{\prime\prime}\circ(\c E^{qi}\circ\c E)\right)\leq
\max_{\c E'}\overline{F}\left(\c E^{\prime}\circ\c E\right)=
\overline{F}\left(\c E^{qi}\circ\c E\right).
\end{equation}
This inequality is saturated if we take $\c E^{\prime\prime}$ equal to the identity map.
\end{proof}
On the other hand, let $\mathfrak{C}_{\c E}$ denote the set of
quantum channels for which $\c E$ defines the quasi-inverse. We argue that such
a set is convex because for any $\c E_1,\c E_2\in \mathfrak{C}_{\c E}$ we
have $p\c E_1+(1-p)\c E_2\in\mathfrak{C}_{\c E}$, where $0\leq p\leq1$.
In this sense, $\mathfrak{C}_I$ is a special convex subset of quantum channels
which are not correctable, i.e. the identity map is its quasi-inverse.
According to the above proposition, applying the quasi-inversion
we actually send a given quantum channel to this special subset since
we cannot correct the fidelity afterwards.\\

As an example consider the tetrahedron of Pauli channels,
$\Phi=\sum_ip_i\Phi_i=
\sum_{i=0}^3 p_i \sigma_i\otimes\overline{\sigma}_i$
where $\sigma_0=I$ and other $\sigma_i$'s are the Pauli matrices.
According to the results of \cite{kbf}, the subset of Pauli channels for which
$\forall i\ \ p_0\geq p_i$ belongs to $\mathfrak{C}_I$. This is also 
a consequence of the result of Example \ref{ortho-unit} in the next section, since $\sigma_i$'s satisfy orthogonality. The set of non-correctable Pauli channels is
presented in  Fig~\ref{tet}.
\begin{figure}[h]
\centering
\includegraphics[scale=0.3]{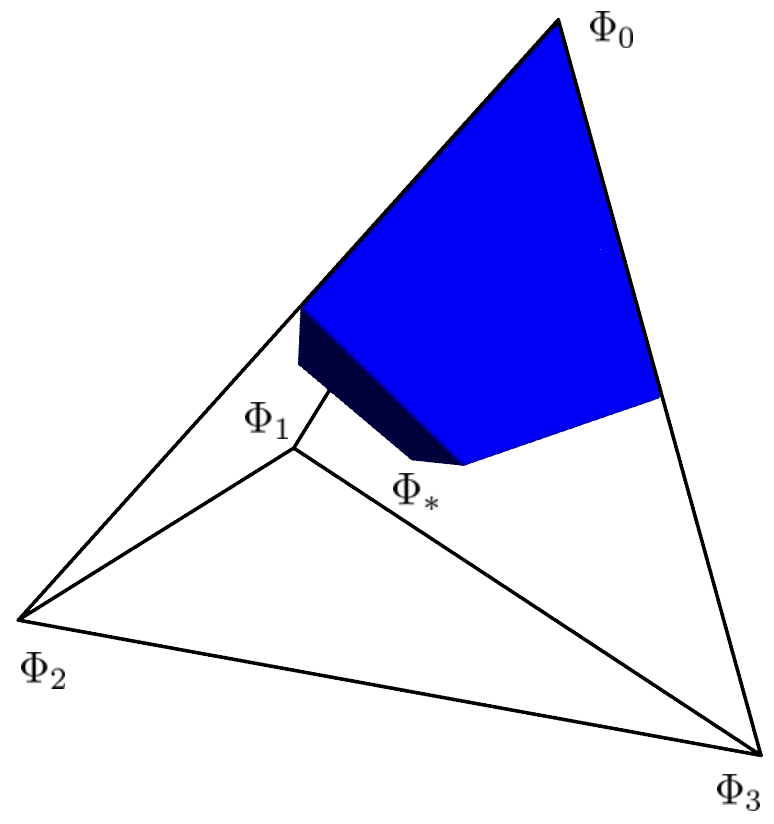}
\caption{The set of Pauli channels. The blue region shows $\mathfrak{C}_I$,
 the convex subset of the channels whose quasi-inversion is the identity map.
$\Phi_*$ is the center of the tetrahedron defined by
$\frac14\sum_{i=0}^3 \Phi_i$.}\label{tet}
\end{figure}

Two unitarily equivalent channels are defined as
$\c E_1$ and $\c E_2$ such that 
$\c E_2=\c E_ U\circ\c E_1\circ\c E_V$ 
where $\c E_ U$ and $\c E_ V$ are two unitary maps.
We use $\Phi_1$ and $\Phi_2$ to denote $\c E_1$ and $\c E_2$. 
Unitary maps, $\c E_U$ and $\c E_V$, are
represented by  $\Phi_{U}$ and $\Phi_V$. Hence,
\begin{eqnarray}\label{unitarily}
\nonumber		\Tr\left(\Phi^{qi}_2\Phi_2\right)&=&
		\max_{\Phi^{\prime\prime}}\Tr\left(\Phi^{\prime\prime}\Phi_2\right)=
		\max_{\Phi^{\prime\prime}}\Tr\left(\Phi^{\prime\prime}
		\Phi_U\Phi_1\Phi_V\right)\\&=&
		\max_{\Phi^{\prime\prime}}\Tr\left(\Phi_V\Phi^{\prime\prime}
		\Phi_U\Phi_1\right)=\max_{\Phi'}\Tr\left(\Phi^{\prime}
		\Phi_1\right)=\Tr\left(\Phi^{qi}_1\Phi_1\right),
	\end{eqnarray}
where we have used the fact that the set of all quantum channels,
$\c C_d$, is invariant under unitary transformations.
This relation proves the quasi-inverse channels of
$\c E_1$ and $\c E_2$ are also unitarily equivalents,
$\c E^{qi}_2=\c E_{V^{-1}}\circ\c E^{qi}_1\circ\c E_{ U^{-1}}$, and
they can reach the same amount of fidelity after correction.
This fact expands the result  of \cite{kbf} related to
unitarily equivalent channels of the form 
$\c E_2=\c E_U\circ\c E_1\circ\c E_{U^{-1}}$.\\

Here a crucial difference between the qubit channels and
higher dimensional channels shows up. In the case of qubits,
a complete characterization of qubit channels exists and it is known that
any qubit channel has the decomposition
${\cal E}={\cal E}_{U}\circ {\cal E}_c\circ {\cal E}_V$, where ${\cal E}_c$
is a canonical map with diagonal distortion matrix
$M_c=\diag (\lambda_1, \lambda_2,\lambda_3)$.
The
 signed singular values \cite{karol2} of the matrix  $M_c$
 confined inside a tetrahedron
$\Delta$ whose extreme points are unitary operations
$\rho\lo \sigma_i\rho \sigma_i^\dagger\ \ \  (i=0,1,2,3)$.
Therefore the task of finding quasi-inverse of any qubit channel is
considerably easy compared with higher dimensional channels
where such a canonical decomposition does not exist and even if there was,
with a presumably diagonal matrix $M$, we were faced with
a highly complex characterization of the vector  $\lambda$.
It is  known that the structure of the convex set of higher
dimensional channels is far more complex than that of a simple tetrahedron,
and in particular it is known that the extreme points of this set are not
necessarily unitary channels. A well-known counter example is the
Landau-Streater channel \cite{LS} which will be discussed
in Section \ref{inverseexamples}.
Let us now put general bounds on the improved average fidelity.

\begin{theorem}\label{upp-loww}
	The average input-output fidelity of a channel after correction has the
	following upper and lower bounds:
\begin{equation}\label{up-low}
\frac{f+1}{d+1}\leq\overline{F}(\c E^{qi}\circ\c E)
\leq\frac{p_m+1}{d+1},
\end{equation}
where
$f=\max_{\ket\beta}\bra\beta C_{\c E}\ket\beta$
is the fully
entangled fraction of the Choi matrix  ${C}_{\cal E}$ of the channel ${\c E}$
and $ |\beta\ra$ denotes a maximally entangled state,
while  $p_m$ is the maximal eigenvalue of ${C}_{\cal E}$.
\end{theorem}
Before proceeding with the proof let us mention in view of Eq.
\eqref{input-output nd entanglement fidelity} and the definition of
fully entangled fraction, the lower bound in above equation is actually
an upper bound for the input-output fidelity before we correct it with
quasi-inverse map, $\overline{F}(\c E)$.
\begin{proof}
To prove the upper bound, we note that 
\begin{eqnarray}
\Tr(\Phi_{{\cal E}'}\Phi_{\cal E})&=&\sum (\Phi_{{\cal E}'})_{ij,kl}(\Phi_{\cal E})_{kl,ij}=\sum (C_{{\cal E}'})_{ik,jl}(C_{\cal E})_{ki,lj}\nonumber\\&=&
\sum (SC_{{\cal E}'}S)_{ki,lj}(C_{\cal E})_{ki,lj}=\Tr(\tilde{C}_{{\cal E}'}C_{\cal E})\nonumber\\ &\leq& p_m \Tr(\tilde{C}_{{\cal E}'})=
p_m\Tr(C_{{\cal E}'})=dp_m,
\end{eqnarray}
where $S=\sum_{i,j}|i,j\ra\la j,i|$ is the swap operator, 
$\tilde{C}_{\c E^\prime}=(SC_{\c E^\prime}S)^\T$,
and $p_m$ is the largest eigenvalue of the Choi matrix 
of the channel ${\cal E}$. In writing these 
equations we have used Eqs. \eqref{reshuffling} and \eqref{reshuff}, and the fact that $C_\c E$ and $\tilde{C}_{\c E^\prime}$ are 
Hermitian and positive matrices.
The above inequality leads to the following upper bound for the improved average fidelity
\be
\overline{F}({\cal E})\leq \frac{p_m+1}{d+1}.
\ee
In order to obtain a lower bound, we note that in view of
(\ref{quasi definition})
\be
\Tr(\tilde{C}_{{\cal E}'}C_{\cal E})=\Tr(\Phi_{{\cal E}'}\Phi_{\cal E})\leq
 \Tr(\Phi_{\cal E}^{qi}\Phi_{\cal E})\ \ \  \forall \ { C}_{{\cal E}'},
\ee
and choose for $\tilde{C}_{{\cal E}'}$ to be equal to $d|\beta\ra\la \beta|$,
where $|\beta\ra$ is a maximally entangled state.  This gives the lower bound.
\end{proof}

One of the main differences between the qubit case and the
higher dimensional channels is that the singular values of $M$
in the qubit case are always less than or equal to one. As a result,
the Bloch vector cannot be stretched by applying the
distortion matrix $M$ in the qubit case, while as we will show in an explicit example,
this is not necessarily the case for higher dimensional channels.
	\begin{figure}[H]
	$${\includegraphics[scale=0.3]{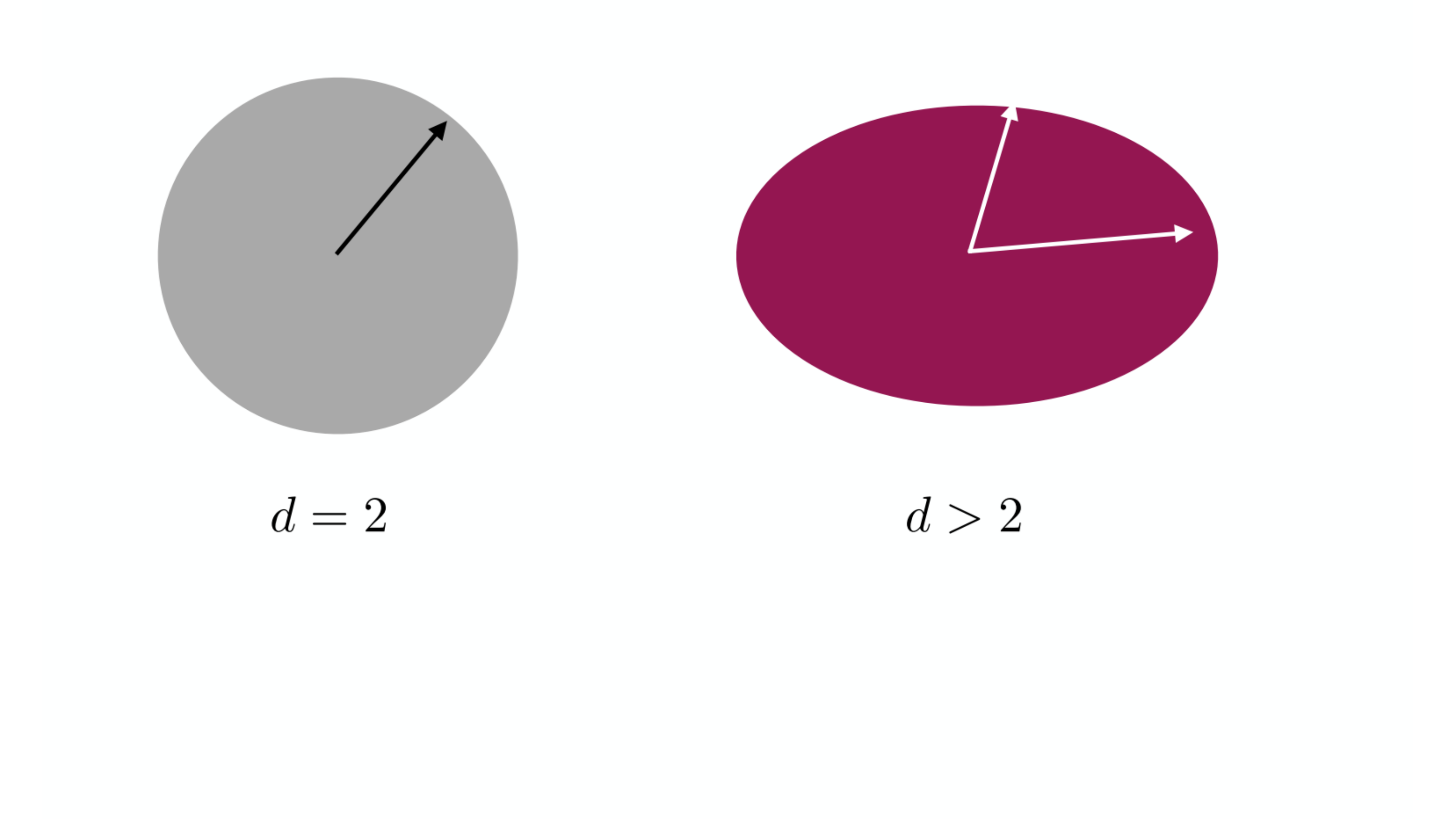}}$$\vspace{-2.5cm}
	\caption{(color online) Left:  In $d=2$ every Bloch vector ${\bf r}$ is rotated or
	shrunk when it is acted on by the distortion matrix of
	a qubit channel. Right:
	 for $d>2$, the distortion matrix of a quantum channel can also stretch
	 the Bloch vector. }
	\label{boundary}
	\end{figure}
There are certain channels whose distortion matrix $M$
in Eq. (\ref{superoperator-affine})
can stretch the generalized Bloch vector {\bf r}. This is due to the
non-spherical shape of the space of quantum states in higher dimensions, see Fig. \ref{boundary}.
We will elaborate on this point and its consequences in Appendix \ref{A2}.

\begin{proposition}
The quasi-inverse of the tensor product of quantum channels is
the tensor product of the quasi-inverses, i.e. if $\c E=\c E_1\otimes\c E_2$
then $\c E^{qi}=\c E_1^{qi}\otimes\c E_2^{qi}$.
\end{proposition}
\begin{proof}
Let $\rho_{AA'BB'}=\rho_{AB}\otimes\rho_{A'B'}$, then one can show
the min-entropy \eqref{min-ent} is additive \cite{KRS09}, i.e.
$H_{\min}(BB'|AA')=H_{\min}(B|A)+H_{\min}(B'|A')$.
The proof of the proposition then becomes straightforward using
Eq.~\eqref{op-min-ent} and the fact that for the tensor product of quantum
channels, the Jamio{\l}kowski state is in the tensor product shape.
\end{proof}
Applying this proposition we show that for
$N$ copies of a given quantum channel in a multipartite setting
\begin{equation}\label{npart}
\max_{\c E'}\overline{F}(\c E'\circ\c E^{\otimes N})=
\overline{F}\left((\c E^{qi})^{\otimes N}\circ\c E^{\otimes N}\right)
\leq\left(\overline{F}(\c E^{qi}\circ\c E)\right)^N.
\end{equation}
To prove this equation we note that for any quantum channel $\c Q$
\be
\Tr \Phi_{{\cal Q}^{\otimes N}}=(\Tr\Phi_{\cal Q})^N.
\ee
Let us apply $\Tr(\Phi_{\cal Q})=x$ for simplicity of notation.
In view of \eref{fbar}, we prove whenever 
$1\leq x\leq d^2$ the following inequality holds:
\be\label{nq}
\overline{F}(\c Q^{\otimes N})=\frac{d^N+x^N}{d^N(d^N+1)}\leq (\frac{d+x}{d(d+1)})^N=\left(\overline{F}(\c Q)\right)^N.
\ee
To prove this relation we note that both sides are increasing functions of $x$. 
 When $x=1$ or $x=d^2$  both sides are equal. Comparing their derivatives with 
 respect to $x$ at $x=1$, we see the right hand side function grows faster at 
 $x=1$ which proves the inequality \eqref{nq}.  To prove Eq. \eqref{npart}, it 
 remains to show that $1\leq \Tr(\Phi_{\c E^{qi}\circ\c E})\leq d^2$ for 
 the composition of any quantum channel and its quasi-inverse. 
 The later is, however, a consequence of Eq. \eqref{up-low} and the fact 
 that  fully entangled fraction of the Choi matrix is greater than $1/d$.\\
  
\section{Examples}
\label{inverseexamples}
Taking into account the  upper and lower bounds given in (\ref{up-low})
and the general theorems of the previous section, in this section
we will  consider a few classes of examples,
as  the optimization problem (\ref{relation of corrected fidelity}) cannot
be solved analytically in the general case.
In the case of single-qubit systems a complete classification of quantum
channels \cite{rus,fuj} leads to an explicit description
of their quasi-inverse which can be unitary \cite{kbf}.
For  higher dimensional channels,  the quasi-inverse may not
necessarily be a unitary map. Thus identifying the quasi-inverse
is related to finding the extreme points of the set  $\c C_d$
of quantum channels, which remains an open problem.

\begin{example}[\bf Mixed unitary channels with orthogonal unitaries]
\label{mixedU}\normalfont\label{ortho-unit}
A mixed unitary channel is defined  \cite{AL07,aud} by
\begin{equation}\label{mixed unitary}
\c E(\rho)=\sum_{\a=1}^{r}q_\alpha V_\a\rho V_\a^\dagger,
\end{equation}
where $\{V_\a\}_{\a=1}^{r}$ is an arbitrary set of
$r$ unitary transformations.
We restrict ourselves to the case where the unitaries are orthogonal
with respect to the Hilbert-Schmidt scalar product,
$\Tr(V_\a^\dagger V_\beta)=d\delta_{\a,\beta}$.
In analogy to the construction of an approximate
time reversal proposed in \cite{karol},
the quasi-inverse of $\c E$ is then the unitary channel
${\cal V}_m^\dagger:\rho\lo V_m^\dagger\rho V_m$,
where $V_m$ corresponds to the largest weight $q_m$ in
the mixture (\ref{mixed unitary}).
To see this  note that the Choi matrix of this channel reads
\be
C_{\cal E}=\sum_\a q_\a|V_\a\ra\la V_\a|,
\ee
where we have used the correspondence (\ref{vectorized}).
In view of the orthogonality of the vectors $|V_\a\ra$,
this is then the spectral decomposition of the Choi matrix $C_{\cal E }$
with eigenvalues equal to $p_\alpha=dq_\alpha$
(note that $\ket{V_\alpha}$ is not normalized).
Let $q_m=p_m/d$ be the largest of coefficients in \eqref{mixed unitary}. In view of the diagonal nature of the Choi matrix, $p_m$ is the largest eigenvalue of the Choi matrix. Moreover, by taking $|\beta\ra=|V_m\ra$, one sees that $f=p_m$. So the upper and lower bounds of \eqref{up-low} coincide and we find the quasi-inverse is 
the unitary map induced by $V_m^\dagger$, i.e.
${\cal E}^{qi}=\c E_{V_m^\dagger}$, with

\be
	\overline{F}({\cal E}^{qi}\circ {\cal E})=\frac{p_m+1}{d+1}.
\ee
\end{example}

Let us emphasize again that this result is valid only for mixture of unitary maps
corresponding to unitary matrices mutually orthogonal
in sense of the Hilbert-Schmidt scalar product.
If this assumption is not satisfied
the quasi-inverse is not the inverse of one of  the unitaries.
As the Example \ref{A3} shows.

\begin{example}[\bf Uniform mixture of orthogonal conjugations]\normalfont
\label{example LS} In the following two examples, we will bring some quantum channels which are self-inverse. 
\begin{theorem}\label{t10}
Let $\c E$ be a unital channel obtained by the uniform mixture of conjugation 
(not necessarily unitary ones) by matrices which are orthogonal to each other.
Such a channel is specified by
\begin{equation}\label{unif.con.}
\c E(\rho)=\sum_{\alpha=1}^qX_{\a}\rho X^\dagger_{\a},
\end{equation}
where 
\begin{equation}
\sum_{\a=1}^qX_{\a}^\dagger X_{\a}=
\sum_{\a=1}^qX_{\a}X_{\a}^\dagger=I_d,\qquad{\rm and}\qquad
\Tr\left(X_{\a}^\dagger X_{\b}\right)=\frac{d}{q}\delta_{\a\b}.
\end{equation}
We now show that the quasi-inverse of this map is given by its dual, i.e. 
$\c E^{qi}(\rho)=\sum X_{\a}^\dagger\rho X_{\a}$.
\end{theorem}
\begin{proof}
The Choi matrix of the channel $\c E$ \eqref{unif.con.} is 
\begin{equation}
C_{\c E}=\sum_{\a}\project{X_{\a}},
\end{equation}
where $\ket{X_{\a}}$ is based on the correspondence of Eq. \eqref{vectorized} and it
fulfills $\bra{X_{\alpha}}X_\beta\rangle=\frac{d}{q}\delta_{\alpha\beta}$.
Thus we find that the above equation is indeed the spectral decomposition 
of the degenerated Choi matrix with the $q$-fold degenerated 
largest eigenvalue equal to $p_m=\frac{d}{q}$. The superoperator of this channel
is given by 
\begin{equation}
\Phi_\c E=\sum_{\a =1}^{q}X_\a\otimes X_\a^\ast.
\end{equation}
Note that since $\c E$ is assumed to be unital, $\Phi^\dagger_{\c E}$ is also a 
valid quantum channel corresponding to the dual of $\c E$. Composing 
$\Phi_\c E$ and $\Phi_\c E^\dagger$, we get
\begin{equation}
\Tr\left(\Phi^\dagger_\c E\Phi_\c E\right)=
\sum_{\a,\b=1}^q|\Tr\left(X_\a^\dagger X_\b\right)|^2=
\sum_{\a=1}^q (\frac{d}{q})^2=\frac{d^2}{q}=dp_m.
\end{equation}
So by such a composition the upper bound of Eq. \eqref{up-low} 
is obtained, which completes the proof. 
\end{proof}

The fidelity after correction for the channel \eqref{unif.con.} 
then reads 
$\overline{F}(\c E^\dagger\circ\c E)=\frac{d+q}{q(d+1)}$, while
for the case with $\Tr(X_\a)=0$ we have
$\overline{F}(\c E)=\frac{1}{d+1}$ before applying quasi-inversion 
which admits significant improvement specially in higher dimensions 
and when $q$ is not large.
Note that for $q=1$ the average fidelity after correction is $1$,
as it is expected.
Moreover, it is obvious that if the operators $X_\a$ are 
Hermitian, $\c E$ is its own quasi-inverse. Two explicit examples for this case
 are provided in what follows.\\
 
 {\bf Example 2.1 (Landau-Streater (LS) channel).}
Consider the Landau-Streater channel  \cite{LS},
\begin{equation}\label{LSchannel}
\c E(\rho)=\frac{1}{j(j+1)}\left(J_1\rho J_1+J_2\rho J_2+
J_3\rho J_3\right),
\end{equation}
where $J_i$ are the Hermitian generators of $SU(2)$ in
its irreducible representation in dimension $d=2j+1$ and they satisfy
$\Tr(J_m^\dagger J_n)  =\frac{1}{3}j(j+1)(2j+1)\delta_{mn}$.
It is clear then that LS channel is a special case of Eq. \eqref{unif.con.} with $q=3$. 
It is worth mentioning that the Landau-Streater channel
 (\ref{LSchannel})
is an extreme point of $\c C_d$ when $d\geq3$ \cite{LS}.
However for  $d=2$,
the Landau-Streater channel is not an extreme point of $\c C_2$,
so in view of Theorem \ref{extreme-ness}
there exist several quasi-inverse channels.\\

{\bf Example 2.2 (Generalized Landau-Streater channel).}
Let $G$ be a  Lie group with dimension $n$, with Lie algebra generators $A_\a$  
where $1\leq\a\leq n$. Let $D_\mu$ be an irreducible unitary representation of 
the Lie algebra. We define the generalized Landau-Streater channel as 
\begin{equation}
\c E_\mu(\rho)=\frac{1}{c_\mu}\sum_{\a=1}^nD_\mu(A_\a)\rho D_\mu(A_\a)^\dagger,
\end{equation}
where $c_\mu$ is the value of the second Casimir operator in this representation
\begin{equation}
c_\mu I=\sum_{\a=1}^n D_\mu(A_\a)^2.
\end{equation}
The generators can be made orthogonal so that
\begin{equation}
\Tr\left[D_\mu(A_\a)D_\mu(A_\b)\right]=\frac{c_\mu d_\mu}{n}.
\end{equation}
Such a channel satisfies the assumption of Theorem \ref{t10} and is hence its 
own inverse.
\end{example}

\begin{example}[\bf The  transverse-depolarizing and depolarizing channel]
\normalfont
The transverse-depolarizing channel is defined as
$\c E_w^{td}(\rho)=(1-w)\rho^\T+\frac{w}{d} \Tr(\rho) I_d$, where
$\frac{d}{d+1}\leq w\leq \frac{d}{d-1}$ to satisfy complete positivity.
The superoperator of this channel is given by:
\begin{equation}\label{td-map}
\Phi_w^{td}=(1-w)S+w\project{\phi^+},
\end{equation}
and the Choi matrix is equal to 
\begin{equation}
C_w^{td}=(1-w)S+\frac{w}{d}I\otimes I.
\end{equation}
$S$ is the swap operator introduced in the proof of Theorem 
\ref{upp-loww}. It is a Hermitian unitary 
so its eigenvalues are $\pm1$. Thus, the largest eigenvalue of
$C_w^{td}$ is 
\begin{equation}
\begin{cases}
p_m=1-w+\frac{w}{d},\quad\quad\frac{d}{d+1}\leq w\leq 1,\\ \\
p_m=w-1+\frac{w}{d},\quad\quad1\leq w\leq\frac{d}{d-1}.
\end{cases}
\end{equation} 
Let $\c E_\pm(\rho)=\frac{1}{d\pm1}\left(\Tr(\rho) I_d\pm\rho^\T\right)$. 
The channel $\c E_+$ is the transverse-depolarizing channel with
$w=\frac{d}{d+1}$ and $\c E_-$, also called Werner-Holevo channel 
\cite{WH02}, is equal to $\c E_w^{td}$ for
$w=\frac{d}{d-1}$. Indeed, any transverse-depolarizing channel 
is a convex combination of $\c E_+$ and $\c E_-$. 
Now it is straightforward to see 
$\overline{F}(\c E_+\circ\c E_w^{td})$ saturates the upper bound 
of \eref{up-low} when $\frac{d}{d+1}\leq w\leq 1$, so $\c E_+$ is
the quasi-inverse on this domain. However, we find in this region 
the equality
$\overline{F}(\c E_+\circ\c E_w^{td})=\overline{F}(\c E_w^{td})$
holds. This implies that for this range of $w$ the transverse-depolarizing
channel is  non-correctable. On the other hand, for $1\leq w\leq\frac{d}{d-1}$, the upper bound of \eref{up-low} 
is achievable by $\c E_-\circ\c E_w^{td}$ confirming 
that $\c E_-$ is the quasi-inverse in this interval and it can indeed
improve the fidelity by $\Delta\overline{F}=\frac{2(w-1)}{d+1}$.
The parameter space of the transverse-depolarizing channel is 
shown in Fig. \ref{td}.\\
\begin{figure}
\centering\vspace{-1cm}
\includegraphics[scale=0.3]{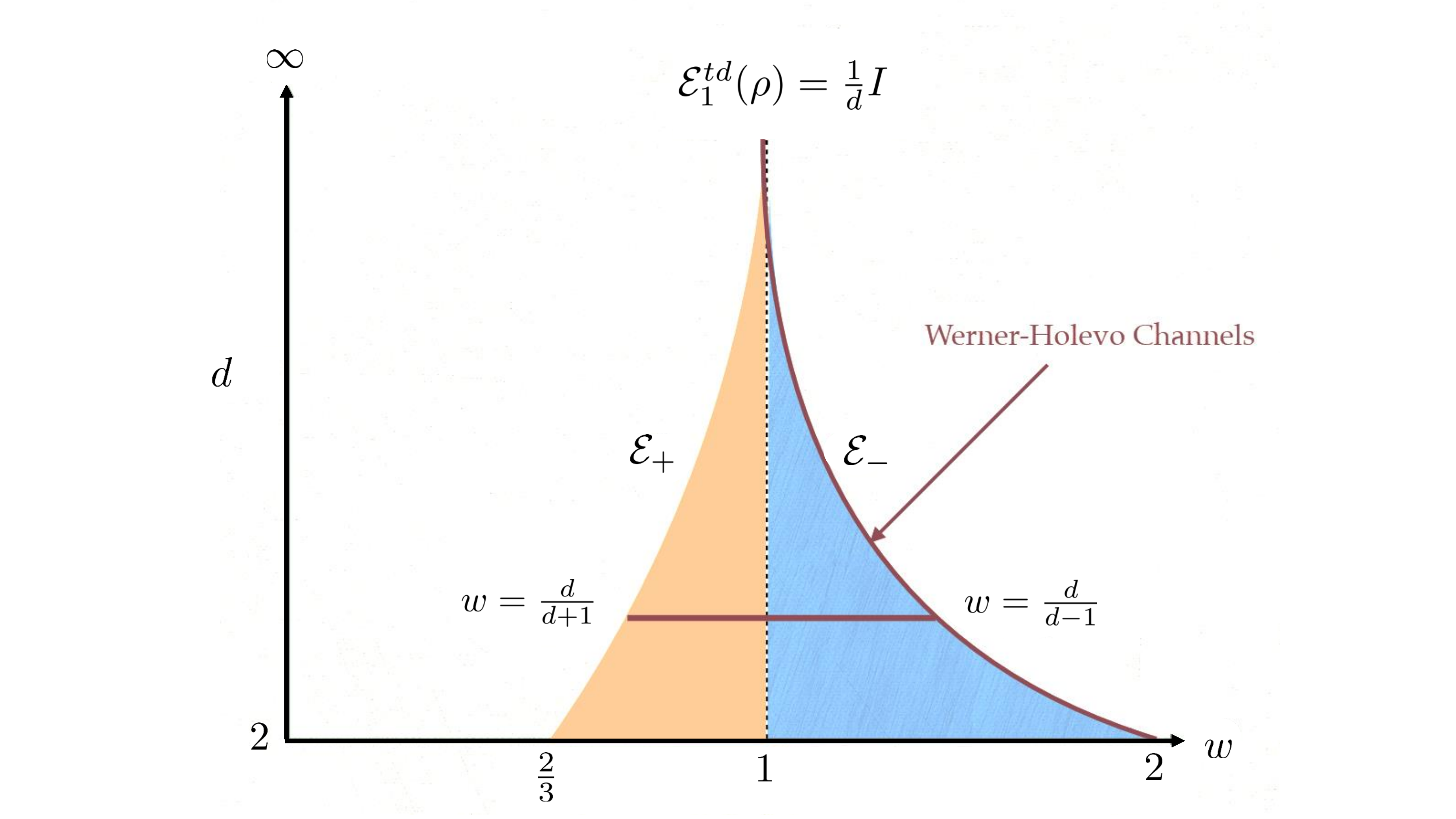}
\caption{(Color Online)A rough sketch of the parameter space of the transverse-depolarizing channel \eqref{td-map}. The colored region defines the
valid range of the parameters for this map to be complete positive. 
The orange part belongs to the set of non-correctable maps, also can be quasi-inverted by $\c E_+$, the boundary of the orange part.
The blue part is the set of correctable channels, invertible by
$\c E_-$, the boundary of the blue part. Fixing $d$, any 
transverse-depolarizing channel is a convex combination of 
$\c E_+$ and $\c E_-$.}\label{td}
\end{figure}

We can also consider the depolarizing channel defined as 
$\c E_q^d=(1-q)\rho+\frac{q}{d}\Tr(\rho)I_d$, where $q$ is a probability.
The superoperator and the Choi matrix are given by
\begin{equation}
\Phi_q^{d}=(1-q)I\otimes I+q\project{\phi^+},\quad
C_q^{d}=d(1-q)\project{\phi^+}+\frac{q}{d}I\otimes I.
\end{equation}
It is now obvious that the largest eigenvalue of the Choi matrix is given by
$p_m=d(1-q)+\frac{q}{d}$ which is actually equal to $\Tr(\Phi_q^{d})/d$. 
This implies that quasi-inverse is the identity map. So
the depolarizing channel is not correctable.
\end{example}

\begin{example}[\bf Covariant Channels]\normalfont
	
A quantum channel ${\cal E}:L(H_d)\lo L(H_d)$ is called {\sl covariant}
with respect to a group $G$, if the following property holds:
\be
	{\cal E}(U(g)\rho U^\dagger(g))=V(g)	{\cal E}(\rho)V^\dagger(g),
	\h \forall \ g \in G.
\ee
in which $U(g)$ and $V(g)$ are two not necessarily equivalent representations
of $g$. From the vectorization (\ref{vectorized}), we find
\be
	\Phi_{{\cal E}}U(g)\otimes U^*(g)=V(g)\otimes V^*(g)\Phi_{\cal E}.
\ee 
This property has implications for the quasi-inverse. To see this we note that
\ba
{\cal E}^{qi}&=&
\underset{\c E'}{\mathrm{argmax}}\  \Tr(\Phi_{\cal E'}\Phi_{\cal E})=
\underset{\c E'}{\mathrm{argmax}}\  \Tr\Big(\Phi_{{\cal E}'}(V^\dagger(g)
\otimes V^T(g))\Phi_{\cal E}(U(g)\otimes U^*(g))\Big)\cr
&=&\underset{\c E'}{\mathrm{argmax}}\   \Big((U(g)\otimes
U^*(g))\Phi_{{\cal E}'}(V^\dagger(g)\otimes V^T(g))\Phi_{\cal E})\Big),
\ea
from which we obtain
\be
\Phi_{{\cal E}^{qi}}V(g)\otimes V^*(g)=U(g)\otimes U^*(g)\Phi_{{\cal E}^{qi}}.
\ee
Equivalently this means that
\be\label{cov1}
{\cal E}^{qi}(V(g)\rho V^\dagger(g))=U(g)	{\cal E}^{qi}(\rho)U^\dagger(g),\h \forall \ g \in G.
\ee
Therefore the quasi-inverse of a covariant channel is also covariant except
that the order of the two representations of the group is reversed.
\end{example}
\begin{example}[\bf mixed unitary channels with commuting unitaries]\normalfont
\label{A3}
In contrast to the qubit case where a classification of quantum channels
facilitates the  study of various aspects of them including their quasi-inverses,
for higher dimensional channels, many aspects do not easily yield an
analytical treatment. In this section we pose a simple problem whose 
solution, as we will see, is quite nontrivial and yet instructive.
We have seen in  Example \ref{mixedU}, that the quasi-inverse of a
 mixed unitary channel of the form
\be\label{pmax}
{\cal E}(\rho)=\sum_k p_k U_k\rho U_k^\dagger
\ee
when $U_k's$ are orthogonal to each other, is the unitary channel
${\cal E}^{qi}(\rho)=U_{max}^\dagger \rho U_{max}$,
where $U_{max}$ is the unitary corresponding to the maximum probability
$p_{max}$ in (\ref{pmax}).  When the unitaries are not orthogonal to
each other, then we know the complete answer only for qubit channels
\cite{kbf}.
For higher dimensional channels we can tackle the simplified version of the
problem, if all
unitary matrices commute, so that they are diagonal in a certain basis.
Under such assumptions we will prove that for qutrit channels
 the quasi-inverse is a unitary map. Our analysis 
 also reveals certain facts about
higher dimensional channels which may be of interest in their own right.
At the end of this section we will consider a concrete case for $d=3$
and the reader can follow the general arguments
 here by looking at that special case.\\

We aim to find
 the quasi-inverse of the channel in (\ref{pmax}) when
$[U_k,U_l]=0\ \ \ \forall \ k$. In the basis in which
all the unitaries are diagonal,
$U_k={\rm diag}(e^{i\theta^{(k)}_1},e^{i\theta^{(k)}_2},\cdots e^{i\theta^{(k)}_d})$,
the superoperator of the channel reads
\be
\Phi = \sum_{k=1}^M p_k U_k\otimes U^*_k
= \sum_{i,j}^d \sum_{k=1}^M p_k e^{i(\theta^{(k)}_i-\theta^{(k)}_j)} |i,j\rangle\langle i,j|
=\sum_{i,j}^d w_{ij} |i,j\rangle\langle i,j|,
\ee
where
\be\label{a}
w_{ij}=	 \langle e^{i\theta_{ij}}\rangle:=\sum_k p_k e^{i(\theta_i^{(k)}-\theta_j^{(k)})}.
\ee
Let $G$ be the unitary group that consists of all diagonal unitary matrices
in this basis.\footnote{This is also known as the phase group.}
Since $U_k$'s commute with any diagonal matrix, the channel
$\mathcal{E}=\{U_k\}$ is $G$-covariant.

\be\label{covv}
U\mathcal{E}(\rho)U^\dagger=\mathcal{E}(U\rho U^\dagger)
\hspace{5mm}\forall\rho,\ \ {\rm and} \ \ \forall\  U\in G.
\ee
According to property (\ref{cov1}),
this implies that the quasi-inverse of this channel is
also $G$-covariant. Expressed in terms of the superoperators,
this means that the superoperator ${\Phi}_{\cal E}$ must commute
with the superoperator of all unitary maps $\rho\lo U\rho U^\dagger$,
where $U\in G,$. The superoperator of the latter is of the diagonal form
\begin{equation}\label{Equality}
\Phi_U=\sum_{mn} e^{i(\phi_m-\phi_n)}|m,n\rangle\langle m,n|    .
\end{equation}
Let the superoperator of the quasi-inverse be given by
\be
\Phi_{{\cal E}^{qi}}=	\sum \Phi_{ij,kl}|ij\rangle\langle kl|.
\ee
Equation (\ref{covv}) now restricts the form of this superoperator to
the following simple form

\be
\Phi_{{\cal E}^{qi}}=	\sum_{i,j}q_{ij}|ij\ra\la ij|+r_{ij}|ii\ra\la jj|.
\ee

The Choi-matrix of the quasi-inverse is obtained by reshuffling
the entries of the superoperator, see \eqref{reshuffling},
which amounts to
\be\label{choiA}
C_{{\cal E}^{qi}}=	\sum_{i,j}r_{ij}|ij\ra\la ij|+q_{ij}|ii\ra\la jj|.
\ee	
According to
$\overline{F}(\c E^{qi}\circ\mathcal{E})=\frac{1}{d+1}\left(
1+\frac{1}{d}\mathrm{\Tr}\Phi_{{\cal E}^{qi}}\Phi_{\cal E}\right)$, 	
the quasi-inverse is the channel which  maximizes the following quantity
\be\label{FF}
F:=	\Tr(	\Phi_{{\cal E}^{qi}}\Phi_{\cal E})=
\sum_{i=1}^d(q_{ii}+r_{ii})+\sum_{i\ne j=1}^d{q_{ij}w_{ij}}.
\ee
Here  we have used the equality $w_{ii}=1\ \forall\ i$,
subject to the condition that $ C_{\c E^{qi}}$ in (\ref{choiA})
designates the Choi matrix of a legitimate quantum channel,
i.e. it is a positive matrix with partial trace equal to the identity
\be\label{choii}
C_{\c E^{qi}}\geq 0\ \ \ \ {\rm and}\ \ \ \Tr_1C_{\c E^{qi}}=I_d.
\ee
In view of the block-diagonal form of the Choi matrix,
it turns out that its eigenvalues are of the form
$\{r_{i\ne j}\}\cup {\rm Eigenvalues\  of\   Y}$,
where $Y$ is the $d$-dimensional matrix
\be
Y=\sum_{i,j=1}q_{ij}|i\ra\la j|+\sum_{i}r_{ii}|i\ra\la i|.
\ee
The second condition in \eqref{choii} leads to the following set of equalities
\be
q_{ii}+\sum_{j}r_{ij}=1\h \forall\ i ,
\ee
which can be rewritten as
\be\label{qqrr}
q_{ii}+r_{ii}+\sum_{j\ne i}r_{ij}=1\h \forall\ i.
\ee	
Note that  $r_{i\ne j}$, being the eigenvalues of the Choi matrix are
non-negative. Therefore in order to maximize the right hand side of
(\ref{FF}),  we can take all of them to be zero, reducing (\ref{qqrr}) to
\be\label{constraint}
q_{ii}+r_{ii}=1\h \forall \ i,
\ee
which further simplifies the expression (\ref{FF}) and reduces our problem
to maximization of the expression
\be
\label{FXS}
F=d+\sum_{i\ne j=1}^d{q_{ij}w_{ij}}\equiv \Tr(X^\T W),
\ee	
subject to the positivity of the following matrix
\be\label{XX}
X=I+\sum_{i\ne j}q_{ij}|i,j\ra\la i,j|.
\ee
The set of all matrices $X$, denoted by $\Omega$ is a convex set.
Therefore the  linear function 	$F$ takes its maximum at the extreme points
of the set $\Omega$.\\

In genera, the problem of finding the extreme points of $d-$dimensional
channels is a difficult and rather non-trivial. It is only known that
the extreme points of the set of  $2$ dimensional unital channels
are unitary maps.
Here we show that for channels defined by the Choi matrix (\ref{choiA}),
 the extreme points are unitary maps if $d\leq 3$.
 We show also a  stronger result:
  for any dimension $d$ the rank of any extreme point
  of this set is less than $ \sqrt{d}$.
To this end, we first need to clarify a few definitions and  a lemma.
In what follows, $\c S$ is a vector space and $\Omega\subset \c S$
is a convex subset of $\c S$.\\

%

A basic property of extreme points of  a convex set  $\Omega$ 
is depicted in figure (\ref{extwitness}). 
In this figure $X_0$ and $X'_0$ are extreme points
while $Y_0$
 is not. We note that any line (no matter how short) passing
through an extreme point like $X_0$ or $X'_0$ contains points which
do not belong to $\Omega$, while for the non-extreme point $Y_0$,
there exists a sufficiently short line (namely the one lying on the edge)
which lies entirely in $\Omega$.
We present  this more formally  in the following statement.
\begin{definition}
Let  $Y_0$ be  a non-extreme point of $\Omega$. Then there is a
$Z\in\c S$, and $\epsilon>0$, such that $Y_0+t(Z-Y_0)\in \Omega$ for all
$t\in (-\epsilon, \epsilon)$. We call $Z$ a witness of non-extremality of $Y_0$.
\end{definition}
This definition implies its equivalent form.
\begin{lemma}
Let $X_0$ be an extreme point of $\Omega$ and let $Z\in \c S$ be an
arbitrary element in $\c S$ such that for all $t\in (-\epsilon,\epsilon)$,
$X_0+t(Z-X_0)\in \c S$. Then $Z=X_0$.
\end{lemma}

\begin{figure}[H]
	$${\includegraphics[scale=0.3]{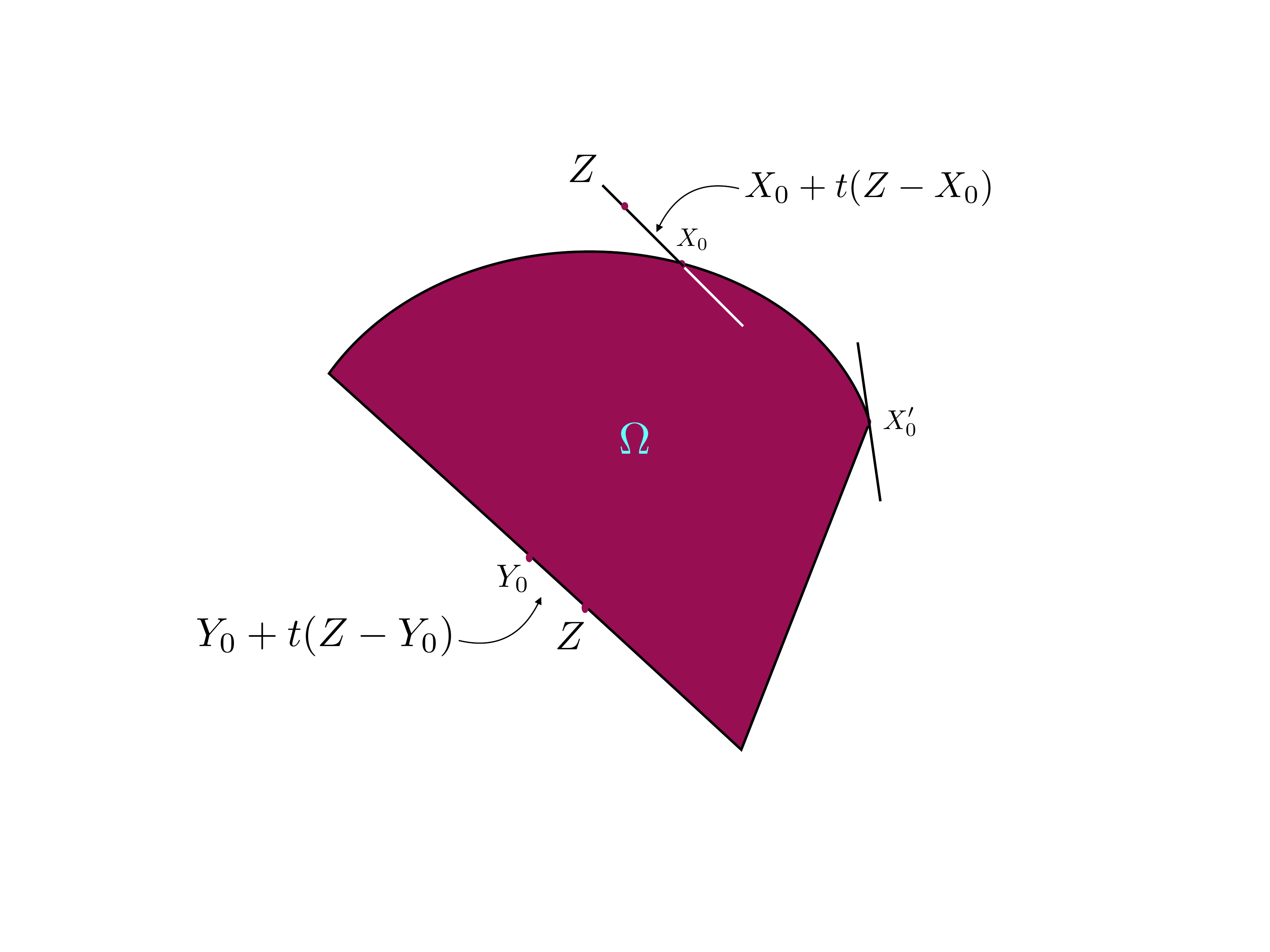}}$$\vspace{-2cm}
	\caption{$X_0$ and $X'_0$ are extreme points of the convex set $\Omega$.
	 Any short segment passing through these points contains points
	 outside $\Omega$. On the other hand $Y_0$ is not an extreme point
	 and there always exists a short segment passing through it which lies
	 entirely inside $\Omega$. }
	\label{extwitness}
\end{figure}
Now we can state and prove
the following result.
\begin{theorem}
Let $\mathbb{V}_m\subset \mathbb{M}_d$ be an $m-$dimensional subspace of the set of
$d-$dimensional complex matrices and let $\Omega\subset \mathbb{V}_m$
be the cone of semi-definite positive matrices. Then the rank $r$ of
any extreme point of $\Omega$ is bounded as  $r\leq \sqrt{d^2-m}$.
\end{theorem}

\begin{proof}
Let $Y_0$ be any element $\Omega$ with  rank $r$. We show that if
$r> \sqrt{d^2-m}$, then we can always find an element $Z\in \Omega$
different from  $Y_0$ such that the sufficiently short line segment
$Y_0+t(Z-Y_0)$ belongs entirely to $\Omega$.
This shows that such points cannot be extreme points of $\Omega$.
To this end, let us expand $Y_0$ in its eigenbasis as
\be
	Y_0=\sum_{\a=1}^r \lambda_\a|u_\a\ra\la u_\a|,
\ee
where $\lambda_\a > 0\ \ \forall \a.$ We choose $Z-Y_0\in \mathbb{M}_d$ in the form
\be
	Z-Y_0=\sum_{\a,\b}z_{\a\beta}|u_\a\ra\la u_\b|,
\ee
and let $t$ be so small that $Y_0+t(Z-Y_0)\geq 0$.  The only other
requirement that is needed for this matrix to belong to $\Omega$ is to satisfy
$d^2-m$ linear homogeneous equations which define the subspace
$\Omega$. This is a system of $d^2-m$ linear homogeneous equations on
$r^2$ variables and if $d^2-m<r^2$, it has always a non-zero solution.
This means that the point $Y_0$ is not an extreme point of the set $\Omega$.
Hence the rank of any extreme point of this set should be less than
or equal to $\sqrt{d^2-m}$.
\end{proof}

\begin{corollary}
As a corollary we find that the extreme points of the set of matrices (\ref{XX})
which is a subset of a $d^2-d$ dimensional space, have rank $r<\sqrt{d}$.
This means that for $d=3$, the extreme points of the set (\ref{XX}) have unit
rank and hence the quasi-inverse of mixed unitary channels with commuting
unitaries is a unitary channel. This theorem by itself does not preclude the
existence of quasi-inverses which are unitary maps in higher dimensions.
\end{corollary}
Consider now the case of $d=3$. Here after setting $r_{ij}=0$ and satisfying
the constraint (\ref{constraint}), the superoperator is a 
diagonal matrix given by 
$\Phi_{{\cal E}^{qi}}=\diag[1,q_{12},q_{21},1,q_{23},q_{31},q_{32},1]$ and its Choi matrix is as follows:
\be C_{{\cal E}^{qi}}=
\left(\begin{array}{ccccccccc}
1&.&.&.&q_{12}&.&.&.&q_{13}\\
.&.&.&.&.&.&.&.&.\\
.&.&.&.&.&.&.&.&.\\
.&.&.&.&.&.&.&.&.\\
q_{21}&.&.&.&1&.&.&.&q_{23}\\
.&.&.&.&.&.&.&.&.\\
.&.&.&.&.&.&.&.&.\\
.&.&.&.&.&.&.&.&.\\
q_{31}&.&.&.&q_{32}&.&.&.&1
\end{array}\right).
\ee
The Hermitian matrices $W$ and $X$,
  which appear in Eq. (\ref{FXS}),
are now $3$ dimensional and take the form
\be
X=\left(\begin{array}{ccc}
1&q_{12}&q_{13}\\ q_{21}&1&q_{23}\\ q_{31}&q_{32}&1
\end{array}\right),\qquad{\rm and}\qquad
W=\left(\begin{array}{ccc}
	1&w_{12}&w_{13}\\ w_{21}&1&w_{23}\\ w_{31}&w_{32}&1
\end{array}\right).
\ee
Having proved that the matrix $X_0$ which maximizes the expression
(\ref{FXS}) is of unit rank, we can write it as $X=|\phi\ra\la\phi |$,
where $|\phi\ra=\left(\begin{array}{ccc}q_1&q_2&q_3\end{array}\right)^\T$
maximizes $\Tr(W|\phi\ra\la\phi|)$,
with $q_i$ being unimodular complex numbers.
 Note that the maximum will be smaller than the largest
eigenvalue of  $W$,  if the corresponding  eigenvectors
 is not built of unimodular entries.
Therefore the quasi-inverse is the unitary map
${\cal E}^{qi}(\rho)=U\rho U^\dagger$, where the unitary operator
$U$ is given by $U=\diag\left(q_1,q_2,q_3\right)$.
In view of (\ref{relation of corrected fidelity}) and (\ref{FXS}),
the final average fidelity becomes
\be
\overline{F}=\frac{1}{4}\bigl(1+\frac{1}{3}\Tr(\Phi_{{\cal E}^{qi}}\Phi_{\cal E})\bigr)
=\frac{1}{4}\bigl(1+\frac{1}{3}max_{|\phi\ra}\la \phi|W|\phi\ra\bigr).
\ee

As a concrete example, consider a spin-1 particle subject to a  magnetic field
in the $z$ direction, where the strength of the magnetic field or the exposure
time is random. The evolution of the state is given by the following channel
\be\label{KooroshChannel}
{\cal E}(\rho)=\int d\tau f(\tau) U(\tau)\rho U^\dagger (\tau),
\h U(\tau)=\left(
\begin{array}{ccc}
e^{i\tau}&&\\ &1&\\ &&e^{-i\tau}\end{array}\right),
\ee
for some distribution $f(\tau)$. For this channel we have

\be
w_{12}=w_{23}=\la e^{i\tau}\ra,\h w_{13}=\la e^{2i\tau}\ra,
\ee
where the matrix $W$ is Hermitian with unit diagonal entries and
the averages are taken with respect to the distribution function $f(\tau)$.
The vector $|\phi\ra$ is then of the form
$|\phi\ra=\left(\begin{array}{ccc}
e^{i\tau_m}&1&e^{-i\tau_m}
\end{array}\right)^\T$,
 where $\tau_m$ is chosen to maximize
$\la \phi|W|\phi\ra$ which is
\be
\la \phi|W|\phi\ra=Re\left[2\la e^{i\tau}\ra e^{-i\tau_m}+
\la e^{2i\tau}\ra e^{-2i\tau_m}\right].
\ee

In view of the form of $X$, the quasi-inverse will be given by
\be\label{KooroshInverse}
V=\left(\begin{array}{ccc}
e^{i\tau_m}&&\\ &1&\\ &&e^{-i\tau_m}
\end{array}\right),
\ee
and the average fidelity after application of the quasi-inverse is given by
\be
\overline{F}=\frac{3+\big\la \left(1+2\cos(\tau-\tau_m)\right)^2\big\ra}{12}.
\ee
Figure \ref{panel}
shows the average fidelity and the increase in average fidelity for a
simple discrete distribution $P(\tau=0)=1-p$ and $P(\tau=\tau_0)=p$.

\begin{figure}[h]\vspace{1cm}
\includegraphics[width=2.8in,height=4.2cm]{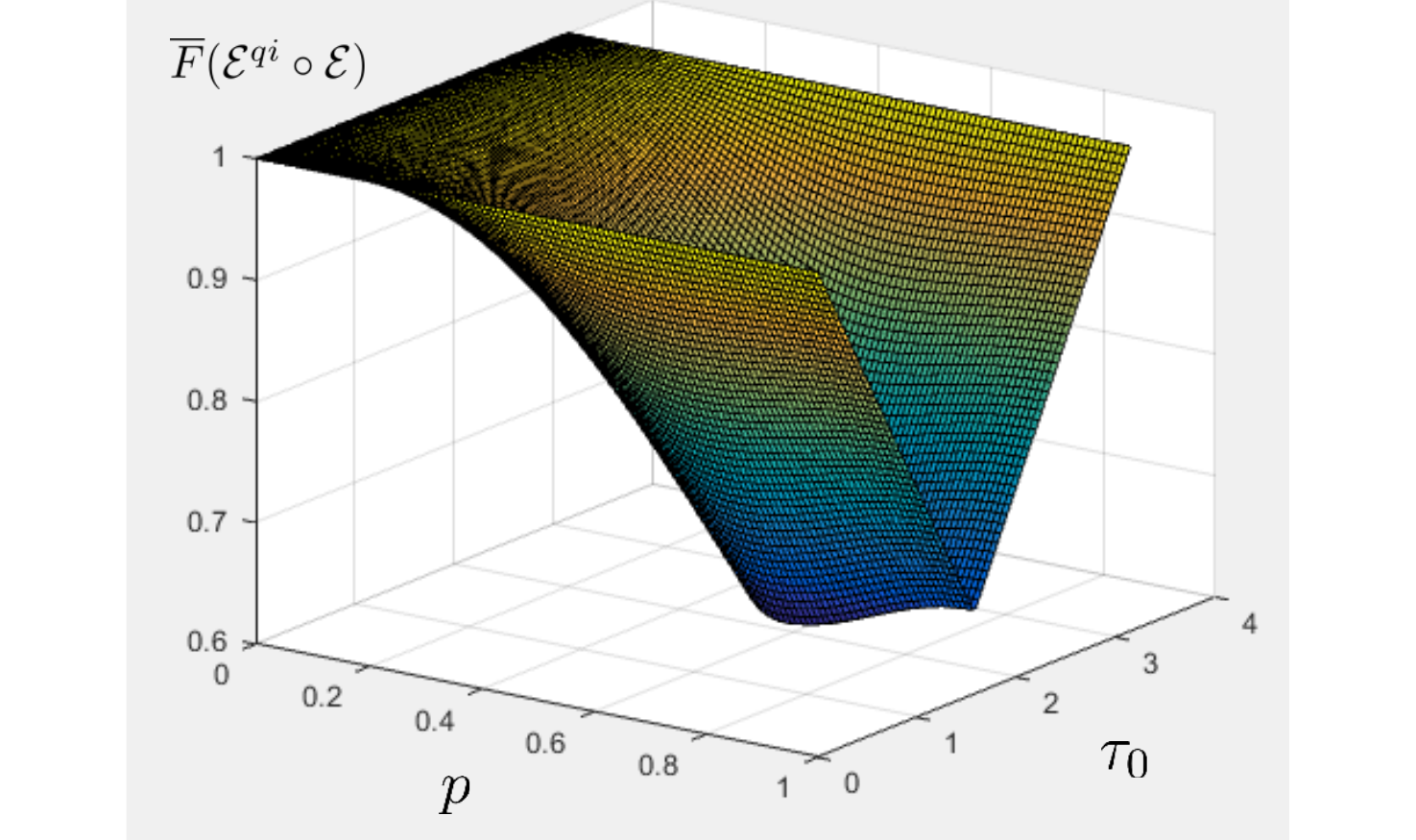}\hfill
\includegraphics[width=2.8in,height=4.2cm]{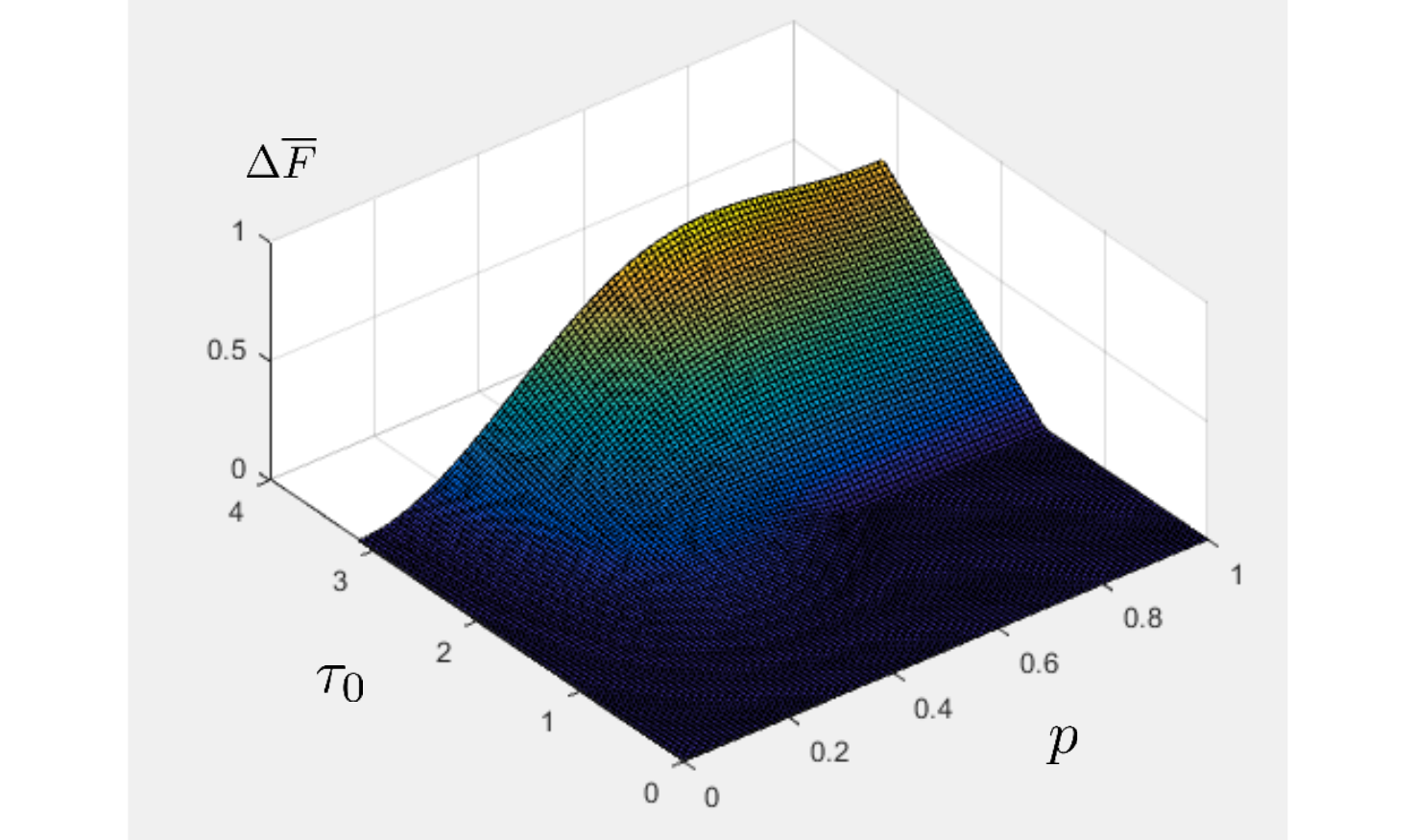}
\caption{(Left) The average fidelity, $\overline{F}$, and 
(Right) the increase of the average fidelity, $\Delta\overline{F}$,
 of the channel  (\ref{KooroshChannel})
	after correction with quasi-inverse (\ref{KooroshInverse}).}\label{panel}
\end{figure}

\end{example}

\section{Quasi-inversion of classical channels}
\label{inverseclassical}
There are several known parallels between probability distributions and
stochastic matrices on the one hand and their  quantum counterparts,
namely density matrices and quantum channels, on the other hand.
For instance, stochastic and bi-stochastic matrices
acting on probability vectors form
classical analogues  of quantum channels and unital quantum channels.
Furthermore, the discrete group of permutations
is the analog of the continuous group of unitary channels.
As the concept of convexity is critical in both domains,
it is instructive to analyze,
how the notion of quasi-inverse works in the
classical setup.\\

 In more concrete terms, the state of classical stochastic system of dimension
 $d$ is a real vector $\vect{p}$ whose entries are non-negative and
add up to one.
The 
set of all probability vectors of length $d$, forms the  simplex
$\Delta_d$ which is a $(d-1)$-dimensional compact and convex set.
For the sake of simplicity, let us use the bra-ket notation here to
mention a probability vector. In this sense, let
$\ket i$ for $i\in\{1,\dots,d\}$ denotes the pure probability vector
whose all components are zero but the $i$-th element which is equal to $1$.
Therefore, a general mixed probability vector
$\vect{p}=(p_1,\dots,p_d)^\T$
can be stated as a convex combination
of $\{\ket i\}$, i.e. $\ket{\vect{p}}=\sum_i p_i\ket i$.\\


A classical channel  is represented  by a stochastic transition
  matrix $T$ of order $d$
  with
non-negative elements where the the sum of all elements in each column is equal to $1$.
This is the analog of trace-preserving property.
The space of stochastic matrices  of order $d$
is a $(d^2-d)$-dimensional convex and compact set which
will be denoted by $\mathcal{S}_d$.\\

Making use of the analogy to the quantum case
consider the generalized Bloch representation  (\ref{density})
of a diagonal density matrix $\rho_{\rm diag}={\rm diag}(\vect{p})$.
Let us order the generators of $SU(d)$ matrices used in (\ref{gamma})
in such a way that $\Gamma_i$  are diagonal  for $i=1,\dots, d-1$.
Then any diagonal matrix $\rho_{\rm diag}$,
representing the classical state $\vect{p}$,
is represented in the Bloch form  (\ref{density}),
where the Bloch vector ${\bf t}_{cl}$
has now only $d-1$ components.\\

Using such a Bloch representation for
point  $\vect{p}$ of the probability simplex  $\Delta_d$
one can represent the action of an arbitrary stochastic matrix
in the form \cite{karol3},
\be
\label{stochastic-affine}
 T=\left(\begin{array}{cc}1&0\\ \sqrt{d-1}\
 {\bf t}_{\rm cl} &M_{cl} \end{array}
 \right).
 \ee
 Note that this representation mimics the Liouville form
(\ref{superoperator-affine}) of a quantum operation,
 with the only difference that the
 classical distortion matrix $M_{cl}$
   forms a $(d-1)$ dimensional truncation
   of the quantum distortion matrix $M$ or order $d^2-1$,
   while the classical  translation vector ${\bf t}_{\rm cl}$
   consists of $d-1$ components of the original
    translation vector ${\bf t}$ of size $d^2-1$.
   Hence in the Bloch basis the classical transition matrix $T$
   forms a block of a matrix representing
   a quantum operation $\Phi_{\cal E}$, which  decoheres to it \cite{kcpz}.
\\

Let us mention explicitly two distinguished classical maps.
The first is the permutation map ${P}_\alpha$ where
${P}_\alpha|i\ra = |\alpha_i\ra$,
where $1\leq i\leq d$, while $\alpha_i$  denote  permutations of $\{1,2,\cdots,d\}$.
The next one is the flat, van der Wearden matrix, denoted by $T_*$,
which sends all input states to the the uniform state.
It implies $(T_*)_{ij}={1/d}$ for any $1\leq i,j\leq d$.
As a result, the summation of elements of the columns of
an assumed matrix $A$ is equal to $1$ if and only if $T_*A=T_*$.
Accordingly, if $A$ is an invertible matrix, the summation of
elements on the columns of ${A}^{-1}$ is also equal to $1$.\\

Finally we note the fidelity of two probability vectors $|\vect{p}\ra$ and
$|\vect{q}\ra$ which is defined as

\be
\la\vect{p}|\vect{q} \ra:=(\sum_{i=1}^d \sqrt{p_iq_i})^2.
\ee
Proceeding in the same way that we did for the quantum case,
here we should define the average fidelity of a classical channel to be
 the fidelity of an output state with the input pure state averaged over
 all input states. Thus we define

\begin{equation}\label{FidelityBar'}
\mathcal{\overline{F}}(T):=
\frac{1}{d}\sum_{i=1}^d F(|i\ra,T|i\ra)=\frac{1}{d}\Tr(T).
\end{equation}

Under these assumptions, the quasi-inversion is the channel
(the stochastic matrix) increasing Eq. (\ref{FidelityBar'}) as much as possible:
\begin{equation}
\label{cl.fid.}
\Tr(T^{qi}T)\; \geq \; \Tr(T' T),
\hspace{4mm} \forall T' \in \mathcal{S}_d.
\end{equation}

To emphasize even further similarity to the quantum case,
consider the Bloch representation
  (\ref{stochastic-affine})
of the classical map $T$ involving
its distortion matrix  $M_{\rm cl}$.
Then the average fidelity of the corrected classical
transformation  $T^{qi} T$
can be expressed as
the maximum over the set of allowed classical distortion matrices,
\begin{equation}
\label{Fidelityclass}
\overline{\c F} ( T^{qi} T)=
\max_{M'_{\rm cl}} \frac{1}{d}
\left(1+ 
\mathrm{Tr} M'_{\rm cl}  M_{\rm cl}\right),
\end{equation}
which is in analogy to Eq.~\eqref{relation of corrected fidelity}.
However, the difference in prefactor with respect to
Eq.~\eqref{relation of corrected fidelity} is due to the averaging over
the set of classical probabilities of a  dimension smaller
than the set of quantum states.\\

All the  arguments of Theorem \ref{boundary point},
based on the linearity  of the fidelity function and its two properties
(\ref{property a}) and (\ref{property b}) are also valid here and therefore
Theorems \ref{boundary point} and \ref{extreme-ness} and the
corollary  \ref{uniqueness} hold true also if we replace the
quantum channels with classical ones.
In particular, the fact that the quasi-inverse of a quantum channel
can be taken to correspond to an extreme point is very important,
since compared with the quantum case, we have a much better knowledge
of the convex set of classical channels and its extreme points.
In the general case, our basic theorem is the following:
\begin{theorem}\label{classical inverse}
Let the stochastic matrix $T$ be such that in each row $i$, 
the element in the $a_i$-th column be the maximum. 
Then its quasi-inverse is found by replacing that single element by $1$ 
and setting all the other elements in that row equal to zero and 
then transposing the matrix.
\end{theorem}

\begin{proof}
Let the matrix $T$ be written as 
$T=\left({\bf t}_1, {\bf t}_2, \cdots, {\bf t}_n\right)^\T$ where 
${\bf t}_i^\T$ denotes the $i-$th row as a vector. 
Denote its quasi-inverse as a matrix 
$T^{qi}=({\bf x}_1, {\bf x}_2,\cdots ,{\bf x}_n)$, 
where ${\bf x}_i$ is the $i-th$ column as a vector. 
The aim is to maximize  $\Tr(TT^{qi})=\sum_i  {\bf t}_i\cdot{\bf x}_i$. 
We can maximize this sum if we maximize each inner product 
$ {\bf t}_i\cdot{\bf x}_i$ independently. 
Since none of the components of the vectors ${\bf x_i}$ can be larger than one, 
the maximization is achieved if we choose each 
${\bf x}_i=(0,0,\cdots 1, 0 , 0 )$ such that $1$ stands 
in the position of maximum component of ${\bf t}_i$. 
This proves the theorem. 
\end{proof}
For example assume a general two dimensional stochastic matrix
parameterized as
\be
T=\left(\begin{array}{cc}1-x& y \\ x& 1-y\end{array}\right),\h 0\leq x,y\leq 1.
\ee
In this case, the quasi-inverse is either $I$, $\sigma_x$, 
or a mixture of these two for $x+y<1$, $x+y>1$ and $x+y=1$, respectively. This is depicted in Fig. \ref{fig3}.
\begin{figure}\vspace{-1.8cm}
	$${\includegraphics[scale=0.3]{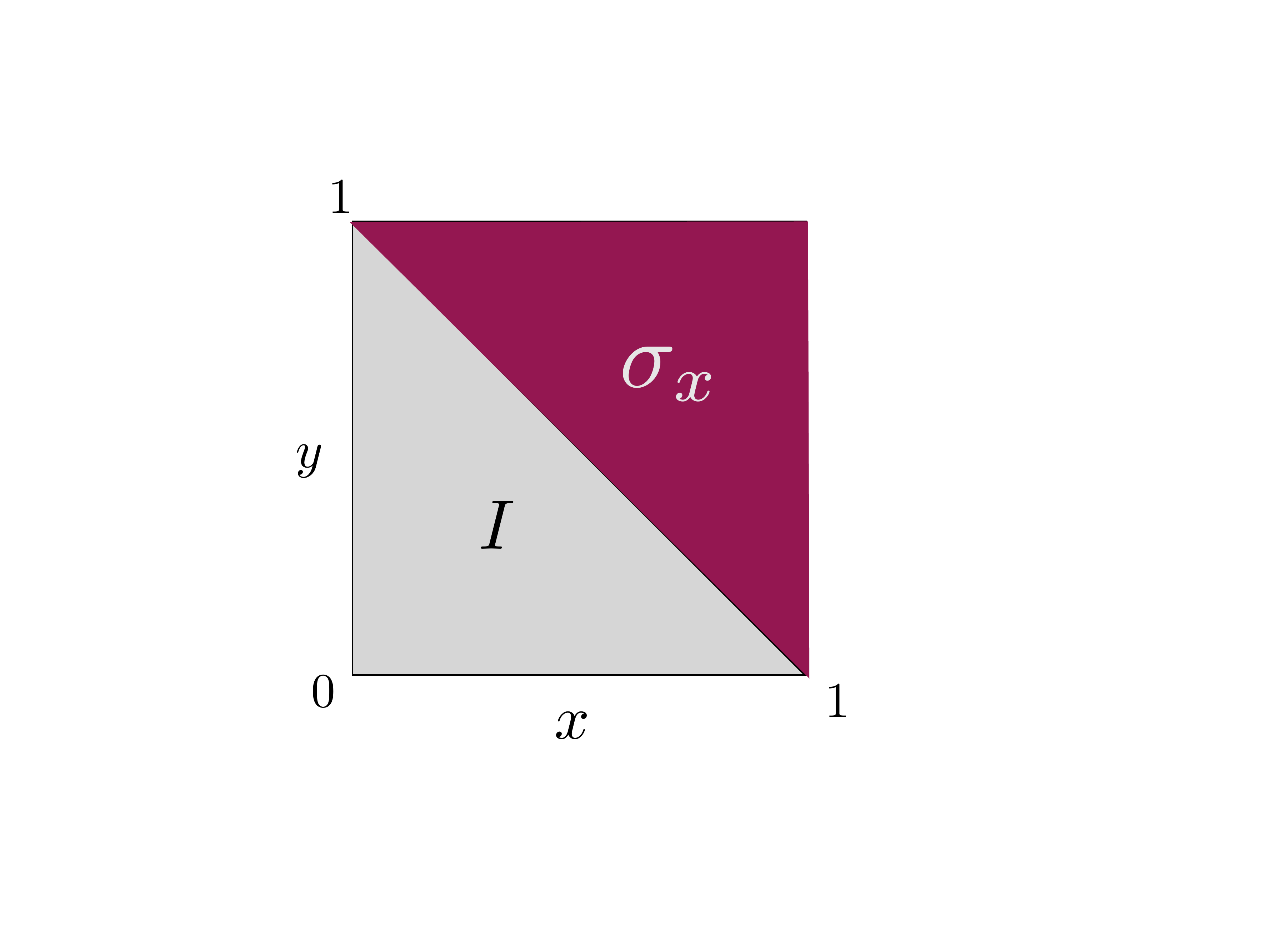}}$$\vspace{-2.5cm}
	\caption{The set of all stochastic matrices of dimension $2$.
	Transitions located in the gray region cannot be corrected
	in sense of input-output fidelity. For other stochastic maps
	one can apply $\sigma_x$ to modify the fidelity.}\label{fig3}
	\label{classic2}
\end{figure}
In higher dimensions, however, more options are possible.\\

\begin{example}
As examples, the quasi-inverse of the stochastic matrices
\be
T_1=\frac18\left(\begin{array}{ccc}
2&6&2\\
4&1&1\\
2&1&5
\end{array}\right),
\h {\rm and} \h	T_2=\frac{1}{48}
\left(\begin{array}{ccc}
12&6&3\\
4&36&42\\
32&6&3
\end{array}\right)
\ee
are given by
\be
T_1^{qi}=\left(\begin{array}{ccc}
0 & 1&0\\ 1&0&0\\ 0&0&1
\end{array}\right),
\h {\rm and} \h  T_2^{qi}=\left(\begin{array}{ccc}
1 & 0&1\\ 0&0&0\\ 0&1&0
\end{array}\right).
\ee
The second example shows that the quasi-inverse of a classical map is not necessarily a permutation.
In the examples above, the average fidelities of the stochastic maps $T_1$
and $T_2$ increase from
 ${1}/{3}$ and ${5}/{18}$ to
${35}/{72}$ and ${43}/{72}$, respectively.   \\

\end{example}

\begin{corollary}
	\label{e6-classical}\normalfont
If $T$ denotes a stochastic matrix for which
$\hspace{2mm} \forall i,j :\ T_{ii} \geq T_{ij}$, then  its quasi-inverse
is $I_d$ which means it is not possible to increase the average  fidelity.
\end{corollary}
As it is seen, these quasi-inverses are at extreme points of the space of
stochastic matrices. To have quasi-inverses which are not necessarily at the
extreme points, we should consider the case where there are more than one
maximum entry in each row. In this case if we follow the argument leading to
Theorem \ref{classical inverse}, we see
the quasi-inverse can be the convex combination of all quasi-inverses which
we construct when we consider only one of these elements.
The next example illustrates this point.

\begin{example}
The stochastic matrix
\be
T=\frac{1}{24}\left(\begin{array}{ccc}
8&3&8\\
12&6&6\\
4&15&10
\end{array}\right)
\ee
has as its quasi-inverse
\be
T^{qi}=\lambda
\left(\begin{array}{ccc}
1 & 1&0\\ 0&0&1\\ 0&0&0
\end{array}\right)
+(1-\lambda)\left(\begin{array}{ccc}
0 & 1&0\\ 0&0&1\\ 1&0&0
\end{array}\right)=\left(\begin{array}{ccc}
\lambda & 1&0\\ 0&0&1\\ 1-\lambda&0&0
\end{array}\right).
\ee
Here the average fidelity increases from the value $\frac{1}{3}$ to
$\frac{35}{72}$ which is expectedly independent from $\lambda$.\\
\end{example}

\begin{theorem}
Among the stochastic matrices with a unique quasi-inverse, 
only symmetric permutations are their own quasi-inverse.
\end{theorem}
\begin{proof}
Uniqueness of the quasi-inverse suggests
that in each row of $T$ there exists a single entry which is 
strictly larger than other elements in the row.
Hence, we have exactly one non-zero array (which is equal to $1$) 
in each column of $T^{qi}$. 
Suppose that the classical channel $T$ is a self-quasi-inverse 
stochastic matrix, i.e. $T^{qi} = T$. 
This equality  implies the existence of 
exactly one non-zero element equal to unity in each row of 
$T^{qi}$. Because we would have had more than one leading value 
in a row of $T$ which violates the uniqueness of 
$T^{qi}$, otherwise. So $T^{qi}$ and consequently $T$ are the same
permutation matrix.
However, a permutation matrix has a real inverse which is equal to 
its transpose. So we have $T^{qi}=T=T^\T$ suggesting that 
$T$ is a symmetric permutation matrix.
\end{proof}

\subsection{On the commutativity of super-decoherence and quasi-inverse.}
 Now that we have discussed the quasi-inverses of both quantum and classical channels, a natural question is whether or not through super-decoherence \cite{kcpz}
 of quasi-inverse of a quantum channel $\c E$, the quasi-inverse of a classical channel $T_\c E$
 can be obtained. In other words, we want to see if the action of taking quasi-inverse commutes with super-decoherence. As we will see below, in general the answer is negative. For convenience we first remind the concept of super-decoherence \cite{kcpz}.\\

 The decoherence channel removes off-diagonal elements of density matrices
and sends any quantum state $\rho$ into its diagonal
$\c D(\rho)=\rho_d=\sum\rho_{ii}\project{i}$, i.e. a classical state.
Defining an analogous process in the space of quantum channels,
one may extract a classical map, a stochastic matrix,
from any quantum channel. This process is called super-decoherence,
noted by $\Delta$, to emphasize that it acts on quantum channels
and not states. For a quantum channel $\c E$,
the assigned classical transition matrix obtained by super-decoherence
is defined by  $T_\c E=\Delta(\c E)$ where
\begin{equation}
(T_\c E)_{ij}=\bra i\c E(\project j)\ket i=(\Phi_\c E)_{ii, jj}=(C_\c E)_{ij, ij}.
\end{equation}
The last equality above shows that this stochastic matrix is actually gained
by decohering the Choi matrix and clarifies why it is called super-decoherence.
It is straightforward to see stochasticity of $T_\c E$ is guaranteed by the
fact that $\c E$ is a positive and trace preserving map.
This is related to the fact that the classical distortion matrix
 $M_{cl}$ of size $(d-1)$ used in (\ref{stochastic-affine})
 forms a block of the quantum distortion matrix $M$ of size $d^2-1$
 present in the Liouville form  (\ref{superoperator-affine})
 of any corresponding  quantum operation \cite{karol3}.\\

Moreover, one can show that if the channel $\c E$ is described by the set of
Kraus operators $\{K_\alpha\}$, then
$T_\c E=\sum K_\alpha\odot K_\alpha^*$ where $\odot$ defines Hadamard
(entry-wise) product.
Through this relation it is easy to see by super-decohering a unital
channel we get a bistochastic matrix, while a unistochastic matrix is obtained
if the input channel is unitary.\\

However, a simple counter example shows that the answer to the question posed at the beginning of this subsection is negative. Consider the case of a single-qubit channel, where the quasi-inverse of a channel is in general a unitary map, $U=e^{i{\theta\bf n}\cdot{\bm\sigma}}$, where $\theta$ and ${\bf n}$ depend on the channel. On the other hand, as shown in Fig. \ref{classic2}, the quasi-inverse of any 2-dimensional classical map is either identity $I$ or the
permutation $\sigma_x$,
which arise due to super-decoherence of a small subset
of the  unitary channels \cite{kcpz19}.

\section{Further results and examples}\label{furtherexamples}

Although Theorem \ref{classical inverse} gives a complete prescription for
finding the quasi-inverse of any stochastic matrix, it is instructive to consider a
few special classes. These examples illustrate further the parallels between
classical and quantum notions of maps and their quasi-inverses.
The first result is the analog of unitarily equivalent quantum maps.\\

Consider two transition matrices  $T_1$ and $T_2$ which are
related by two arbitrary permutations $P_\alpha$
and $P_\beta$ in the following way
$T_2 = P_\alpha \circ T_1 \circ P_\beta$.
Let us call them permutationaly equivalent stochastic matrices.
Pursuing the same approach adopted in obtaining Eq. \eqref{unitarily},
and noting that the set of stochastic matrices, $\c S_d$, 
is invariant under permutations, we get
\begin{equation}
	T_2^{qi} = P^{-1}_\beta \circ T^{qi}_1 \circ P^{-1}_\alpha.
\end{equation}
Moreover, $T_1$ and $T_2$ have the same amount of fidelity after correction.
The next example illustrates the connection with its quantum analog,
Example \ref{mixedU}.\\

\begin{example}[\bf Convex combination of orthogonal permutations]\label{e2}\normalfont
\ A bi-stochastic matrix $T$ is a stochastic matrix with the extra property that
the sum of entries of each row equals unity. It is a well-known result due to
Birkhoff's theorem that any such matrix can be written as a convex
combination of permutations. For dimension $d$, there are $d!$ such
permutations which form the extreme points of the convex set of these
matrices, conventionally called the Birkhoff polytope. Note however that a
bi-stochastic matrix has $(d-1)^2$ independent parameters and hence the convex
decomposition of an arbitrary bi-stochastic matrix in terms of $d!$
permutations is not unique.   Consider now a special class of  bi-stochastic
matrices which are convex combination of orthogonal permutations.   These
permutations are defined by the property that
\be\label{orthogonality of permutations}
	\Tr(P^{-1}_mP_n)=d\delta_{n,m}.
\ee
Equation (\ref{orthogonality of permutations}) indicates that
\be\label{pmpn}
	\forall \ \ k\h P_m|k\ra\ne P_n|k\ra,\ \ \ if\ \ \ m\ne n.
\ee
This means that there are at most $d$ orthogonal permutations in the group $S_d$ of all
permutations of $d$ objects. Furthermore, it implies  that

\be\label{sumP}
	\la j|\sum_mP_m|k\ra=0,1\h \forall\ \  j, \  k.
\ee
Here the sum runs over all the permutations in the convex combination, which may be a subset of all the $d$ orthogonal permutations.
A simple example consists of the set of permutations of the form $\{P_m=X^m,\ \ m=0,\cdots d-1\}$, where $X=\sum_{n=0}^{d-1}|n+1\ra\la n|$ is the full-cycle permutation, or any subset thereof. By definition we have
$P_m|k\ra=|k+m\ra$ which clearly satisfies (\ref{sumP}) as $\sum_m\la j|P_m|k\ra=\sum_m\delta_{j,k+m}=0, 1.$
In this example all the permutations commute with each other. As an example consisting of non-commuting but orthogonal permutations consider the following:
\be
	P_1=\left(\begin{array}{ccc}1&2&3\\ 2&1&3\end{array}\right),\h
	P_2=\left(\begin{array}{ccc}1&2&3\\ 3&2&1\end{array}\right),\h
	P_3=\left(\begin{array}{ccc}1&2&3\\ 3&1&2\end{array}\right),
\ee
with matrix representations
\be
		P_1=\left(\begin{array}{ccc}
		0&1&0\\ 1&0&0\\ 0 &0&1\end{array}\right),\h
		P_2=\left(\begin{array}{ccc}0&0&1\\
		0&1&0\\1&0&0\end{array}\right),\h
		P_3=\left(\begin{array}{ccc}
		0&1&0\\ 0&0&1\\ 1&0&0\end{array}\right).
\ee
One can see that $P_1$ and $P_2$ are orthogonal while $P_1$ and
$P_3$ are not. Also one can see that
\be
	P_1+P_2=\left(\begin{array}{ccc}
	0&1&1\\ 1&1&0\\ 1 &0&1\end{array}\right),
\ee
satisfying the condition (\ref{sumP}) while
\be
		P_1+P_3=\left(\begin{array}{ccc}
		0&2&0\\ 1&0&1\\ 1 &0&1\end{array}\right),
\ee
 violates (\ref{sumP}).  Furthermore one can check that in the group
 $S_3$ with generators
 $\sigma_1=\left(\begin{array}{ccc}
 1&2&3\\ 2&1&3\end{array}\right)$ and
 $\sigma_2=\left(\begin{array}{ccc}
 1&2&3\\ 1&3&2\end{array}\right)$,
 each of the two sets of even and odd permutations, respectively given by
 $(I,\ \sigma_1\sigma_2,\ \sigma_2\sigma_1)$  and
 $(\sigma_1,\ \sigma_2,\ \sigma_1\sigma_2\sigma_1)$
 comprise orthogonal permutations. \\

 Consider now a bi-stochastic matrix of the form
\begin{equation}
	B=\sum \lambda_m P_{m},
\end{equation}
where $\{P_m\}$ are orthogonal.
In what follows we will show the quasi-inverse of any such bi-stochastic
matrix is	$P^{-1}_{m_0}$ where $m_0$ refers to the index
of the greatest coefficient $\lambda_m$.
To see this we note that
\begin{eqnarray}\label{Fbar,Tr(AB)}
	\mathcal{\overline{F}}(T\circ B) &=& \frac{\Tr(TB)}{d}=
	\frac{\sum_m\lambda_m \Tr(T P_{m})}{d} \cr&\leq&
	\frac{\lambda_{m_0} \Tr(T \sum_m P_{m} )}{d}=
	\frac{\lambda_{m_0} \sum_{k,j} \la k|T|j\ra\la j|\sum_m P_{m}|k\ra}{d}.
\end{eqnarray}
Using Eq. (\ref{sumP}) for orthogonal permutations, we find
\begin{equation}
	\mathcal{\overline{F}}(T\circ B)\leq
	\frac{\lambda_{m_0} \sum_{k,j} \la k|T|j\ra}{d} = \lambda_{m_0}.
\end{equation}
where we used the fact that for every stochastic matrix,
the sum of its elements is equal to the dimension $d$.
Now, it is straightforward to see $T^{qi}=P^{-1}_{m_0}$
recovers the average fidelity of $T$ to this upper bound.
\vspace{0.5cm}

Thus the quasi-inverse of convex combination of orthogonal permutations is
the inverse of the single permutation which has the largest share in the convex
combination. This is however not the case if the permutations are not
orthogonal. This is shown in the next example.\\
\end{example}
\noindent
\begin{example}[\bf Convex combination of non-orthogonal permutations]\label{e3}\normalfont
\ Consider the following permutations
\be
	P_1=I\oplus I\oplus I,\ \ \ \ P_2=\sigma_x\oplus \sigma_x\oplus I,
	\ \ \ \ P_3=\sigma_x\oplus I\oplus \sigma_x,
\ee
where $I=\left(\begin{array}{cc}1&0\\ 0 & 1\end{array}\right)$ and
$\sigma_x=\left(\begin{array}{cc}0&1\\ 1 & 0\end{array}\right)$.
These permutations are obviously non-orthogonal in the sense that
$\Tr(P_iP_j^\T)\ne 0.$ Consider now the following convex combination
\be
	P=\lambda_1P_1+\lambda_2P_2+\lambda_3P_3,
\ee
where $\lambda_1+\lambda_2+\lambda_3=1$ and all
$ \lambda_i<\frac{1}{2}.$
In explicit form this permutation matrix is given by
\be
	P=\left(\begin{array}{cccccc}
	\lambda_1 & \lambda_2+\lambda_3&.&.&.&.
	\\ \lambda_2+\lambda_3& \lambda_1&.&.&.&.
	\\ .&.&\lambda_1+\lambda_3& \lambda_2&.&.\\.&.&
	\lambda_2&\lambda_1+\lambda_3&.&.\\ .&.&.&.&
	\lambda_1+\lambda_2&\lambda_3\\ .&.&.&.&
	\lambda_3&\lambda_1+\lambda_2
	\end{array}\right).
\ee
Since $\lambda_i<\frac{1}{2}$ for all $i$, we find
$\lambda_i+\lambda_j>\frac{1}{2}$ for all $(i\ne j)$
(due to the requirement that $\lambda_1+\lambda_2+\lambda_3=1$).
Thus according to Theorem \ref{classical inverse}, the quasi-inverse of
$P$ is given by replacing the largest entry of each row with unity and then
replacing the resulting matrix, hence  $P^{qi}=\sigma_x\oplus I\oplus I$
which is not equal to inverse of any of the permutations
$P_1$, $P_2$ or $P_3$.
\end{example}

\noindent

\section{Concluding remarks}
We have extended the concept of quasi-inversion of qubit channels \cite{kbf} to quantum channels in arbitrary dimensions and to the classical domain, i.e. Markov processes in discrete time. In both cases, a quasi-inverse is a map which when combined with the original channel increases the average input-output fidelity in an optimal way. While
the complete classification of  qubit channels \cite{rus,fuj}, makes a complete characterization of their quasi-inverse possible, in higher dimensions the lack of such classification makes the problem a highly non-trivial one. The most notable difference is that in the qubit case, the extreme points of all unital channels are the unitary maps while in higher dimensions this is not the case any more and no general theorem is known on extreme points. Therefore in this paper, we have 
established  certain general theorems on the nature of the quasi-inverse, and have provided certain bounds on the average fidelity after quasi-inversion. Applying these general results to some concrete cases, we have found in examples 1 through 5  the quasi-inverse for a large class of quantum channels.  Moreover, we have done a  parallel analysis for the classical channels, represented by stochastic matrices. \\

As shown in the Appendix \ref{statis-cl}, we have also obtained exact expressions for the average input-output fidelity for classical channels in any dimension $d$ and have shown  that the quasi-inversion increases this quantity
 from $d^{-1}$ to, approximately, $d^{-1/2}$.
In the quantum case,
we analyzed numerically in  Appendix \ref{statis-qu}
the improvement of the average fidelity of a random channel after correction
with quasi-inversion and after applications of
the best possible unitary evolution -- see Fig. \ref{CorrectedF}.
 If the dimension of the systems increases,
 the non-unitarity of quasi-inversion becomes larger.
 In other words, for a generic channel acting in higher dimensions,
 its  quasi-inverse is usually  non-unitary.\\

It is noteworthy to mention that by applying quasi-inversion,
we aim to get as close as possible to the identity map
in the sense of the maximizing average fidelity
and so by the reduction of the average Bures distance.
We do not expect the notion of quasi-inversion to improve some other
properties of channels, for which the identity map is not the
most distinguished channel.
There are other figures of merit, like the average output purity of a channel
\be
\overline{P}(\c E):=\int{\rm d}\psi\ \Tr\left[ (\c E(\project\psi)\right)^2] ,
\ee
or the unitarity of a channel  \cite{flamina}

\be
\overline{u}(\c E):=\frac{d}{d-1}\int{\rm d}\psi\
\Tr\left[\c E(\project\psi)-\c E(\frac{I}{d})\right]^2,
\ee
which can also be studied in the same way, leading to a different version of the
quasi-inverse, both in the quantum and the classical domain.
We are aware of examples which show that these notions of quasi-inverse are different.  That is, a quasi-inverse which increases the average fidelity can increase, decrease or keep constant the unitarity of a quantum channel. One can also study the effect of quasi-inversion on the cohering power of quantum channels as defined in \cite{ple, mk, zan}.\\

Finally,  an
interesting by-product of our study is Appendix \ref{A2}, where we have constructed special types of channels with affine maps $( M,\vect{ t})$, where the distortion matrix $M$  can stretch the Bloch vectors. Such channels are possible only in dimensions higher than two. These are of course different from channels of the form ${\cal E}(\rho)=|\psi\ra\la\psi|$, where the distortion matrix $M$ vanishes and the stretching is due only to the translation vector $\vect{t}$. \\

We hope that this study can be pursued in different directions, i.e. in obtaining more information about the extreme points of the space of quantum channels, and hence the quasi-inverse of  larger classes of channels, those channels which do not have unique quasi-inverses, and those which are their own quasi-inverse. And also more importantly in finding connections with the recovery maps \cite{ren1,ren2,ren3,renal}.

\bigskip

It is a pleasure to thank Seyed Javad Akhtarshenas, Erik Aurell,
Giulio Chiribella, Sergey Filippov,
Kamil Korzekwa, and {\L}ukasz Pawela
for several discussions and helpful remarks.
This research was partially supported by the grant number G98024071 from Iran National Science Foundation.
Financial support by Narodowe Centrum Nauki under the grant number
DEC-2015/18/A/ST2/00274  and by the Foundation
for Polish Science under the Team-Net NTQC project is gratefully acknowledged.



\appendix
\section{An explicit basis for $L(H_d)$}\label{app1}
In this appendix we choose an explicit representation for the matrices
$\Gamma_i$ belonging to $L(H_d)$. Let $E_{ij}:=|i\ra\la j|$ be the standard basis of matrices,
then the set
\be
H_k=\sqrt{\frac{d(d-1)}{k(k+1)}}(\sum_{i=1}^{k}E_{ii}-kE_{k+1,k+1})
\h k=1,\cdots d-1
\ee
and
\be
X_{ij}=\sqrt{\frac{d(d-1)}{2}}(E_{ij}+E_{ji})\h Y_{ij}=
\sqrt{\frac{d(d-1)}{2}}(-iE_{ij}+iE_{ji})
\ee
form   a basis of Hermitian matrices with the proper normalization in
(\ref{gamma}). They are nothing but the standard Gell-Mann matrices,
properly normalized to satisfy (\ref{gamma}).
To work in parallel with the Bloch representation of any
classical probability vector $\bf p$
it is convenient to order the matrix basis in such a way that
the  $d-1$ diagonal matrices appear first,
$\Gamma_k = H_k$, for $k=1,\dots, d-1$ --
see Section \ref{inverseclassical}.\\

A density matrix can be written as
$\rho=\frac{1}{d}(I+{\bf r}\cdot{\bm \Gamma})$, where
${\vect r}$ is the generalized Bloch vector and 
${\bm \Gamma}$ is a vector constructed by the elements $\Gamma_i$.
We note, as stated after Eq. (\ref{rhosquared}),
that any pure state corresponds to  unit vector ${\bf r}$.
However the converse is not true.
In fact one can easily verify that a matrix like
$\frac{1}{d}(I-H_{d-1})=|d\ra\la d|$ is a pure state while a state like
$\frac{1}{d}(I+H_{d-1})$ is not a state at all,
since it has negative eigenvalues.
Both correspond to antipodal points on the sphere $S_{d^2-2}$.

\section{Some details on the calculations of average input-output fidelity }
\label{A1}
As shown in Eq. (\ref{fidel}), the input-output fidelity is given by
$\overline{F}({\cal E})=\Tr(L\Phi_{\cal E})$, where $L$ is given by
\begin{eqnarray}
\nonumber L&=&
\int{\rm d}\psi\ \ket\psi\bra\psi\otimes\ket{\psi^\ast}\bra{\psi^\ast}\\&=&
\int{\rm d}U\ (U\otimes U^\ast)\project{00}(U\otimes U^\ast)^\dagger.
\end{eqnarray}
The last line shows $L=T\left(\project{00}\right)$ is an isotropic
state  obtained by twirling the state $\ket{00}$. In general for any state $\rho$, its twirling gives
 \cite{HH99}:
\begin{eqnarray}
T(\rho)&=&\int{\rm d}U\ (U\otimes U^\ast)\rho(U\otimes U^\ast)^\dagger
\nn\\&=&
\frac{d^2}{d^2-1}\left((1-f_\rho)\frac{I}{d^2}+(f_\rho-\frac{1}{d^2})P_+\right),
\end{eqnarray}
where $f_\rho=\Tr\left(P_+\rho\right)$ and
 $P_+=\project{\phi_+}$
 is the maximally entangled state. So for the state $\project{00}$ of
 a  $d$-dimensional system, one has:
 \begin{equation}
f_{\ket{00}}=\Tr\left(P_+\project{00}\right)=\frac{1}{d},
 \end{equation}
and
\begin{equation}
L=T\left(\project{00}\right)=\frac{1}{d(d+1)} I+\frac{1}{(d+1)}P_+,
\end{equation}
which coincides with Eq. (\ref{isostate}) in the text. \\

One can also write the average fidelity (\ref{fidel}) in the form
\be
\overline{F}({\cal E})=\Tr(\Lambda \Psi_{\cal E}),
\ee
where
$\Psi_{\cal E}=\sum_\a K_\a\otimes K_\a^\dagger$ and
\begin{eqnarray}
\nonumber \Lambda=
\int{\rm d}\psi\ \ket\psi\bra\psi\otimes\ket{\psi}\bra{\psi}
\end{eqnarray}
is a symmetric Werner state
\be
\Lambda=\frac{1}{d(d+1)}(I+P_{sym})=\frac{1}{d(d+1)}(I+\sum_{i,j}|i,j\ra\la j,i|).
\ee
This will also lead to the same result as in (\ref{fbar}). \\

Finally let us calculate the average fidelity in yet another way,
by expressing the pure states  as
\be\label{pure}
|\psi\ra\la \psi|=\frac{1}{d}(I+{\bf n}\cdot \bm\Gamma).
\ee
From
\be
{\cal E}(|\psi\ra\la \psi|)=\frac{1}{d}
\bigl(I+(M{\bf n}+{\bf t})\cdot \bm\Gamma \bigr).
\ee
and Eq. (\ref{pure}),
one finds
\be
\la \psi|{\cal E}(|\psi\ra\la\psi|)|\psi\ra=\frac{1}{d}\bigl(1
+(d-1){\bf n}\cdot (M{\bf n}+{\bf t})\bigr).
\ee
If we now assume a uniform distribution of pure states on the sphere $S_{d^2-2}$ (which is of course not dense) and hence use the relations
\be\label{integraln}
\int d{\bf n}\ n_i=0,\h \int d{\bf n}\ n_in_j=\frac{1}{(d^2-1)}\delta_{ij}
\ee
we arrive at
\be
\overline{F}({\cal E})=\frac{1}{d}\Bigl( 1+\frac{1}{d+1}\Tr M\Bigr),
\ee
which coincides with the result (\ref{fbar}) which we obtained by integrating
over the invariant  volume $d\psi$ of all pure states. Here we have not used
any specific measure of volume over the sphere $S_{d^2-2}$ and only have
assumed that  pure states are distributed symmetrically (albeit in a parse way)
on this sphere.  \\

\section{Quantum channels which stretch the Bloch vector }
\label{A2}
In two dimension, to satisfy the positivity condition, the affine matrix $M$ of
any positive map, and thus any quantum channel,  fulfills $M^{\T}M\leq I$.
As a result, the Bloch vector corresponding to a qubit state is always shrunk
when $M$ acts on it. Here we show that this is no longer the case in higher
dimensions. This is one of the strange or un-expected properties
of higher dimensional channels which makes the study of their quasi-inverses, among other things, difficult.
As describing this class of quantum maps
is not straightforward we construct an example of
such a channel in this Appendix.\\

 Let $H_d$ be a $d-$dimensional Hilbert space and let $P_1, P_2, Q_1 $ and
$Q_2$  be orthogonal projectors of rank $d_1, d_2, m_1$ and $m_2$
respectively, with ($d=d_1+d_2=m_1+m_2$), i.e.
$$P_1+P_2=Q_1+Q_2=I, \quad P_1P_2=Q_1Q_2=0.$$

The following operators are Hermitian and traceless

\be\label{A and B}
A:=\frac{m_2Q_1-m_1Q_2}{\sqrt{m_1m_2}},\h B=\frac{d_2P_1-d_1P_2}{\sqrt{d_1d_2}},
\ee
and we have
\be
\Tr(AB)=0,\h \Tr(A^2)=d,\h \Tr(B^2)=d.
\ee
Consider now the following measure-and-prepare CPT map

\be
{\cal E}(\rho)=\frac{1}{d_1}\Tr(Q_1\rho)P_1+\frac{1}{d_2}\Tr(Q_2\rho)P_2.
\ee
By expanding the projectors as 

\ba
&&P_1=\sum_{i=1}^{d_1}|i\ra\la i|,\hspace{1.72cm}   P_2=
\sum_{i=d_1+1}^{d}|i\ra\la i|,\cr  &&Q_1=
\sum_{\a=1}^{m_1}|\varphi_\a\ra\la \varphi_\a|,\h   Q_2=
\sum_{\a=m_1+1}^{d}|\varphi_\a\ra\la \varphi_\a|.
\ea
a set of  Kraus operators for this map is given by
\be
K_{i,\a}:=\frac{1}{\sqrt{d_1}}|i\ra\la \varphi_\a| \h i=1,\dots, d_1, \ \ \ \a=1,\dots, m_1
\ee
and
\be
\h L_{i,\a}:=\frac{1}{\sqrt{d_2}}|i\ra\la \varphi_\a| \h i=d_1+1,\dots, d, \ \ \ \a=\m_1+1,\dots, d.
\ee

Note that
\be
\sum_{i,\a}K_{i,\a}^\dagger K_{i,\a}+\sum_{i,\a}L_{i,\a}^\dagger L_{i,\a}=I_d,
\ee
but
\be
\sum_{i,\a}K_{i,\a} K_{i,\a}^\dagger+\sum_{i,\a}L_{i,\a} L_{i,\a}^\dagger=\frac{m_1}{d_1}P_1+\frac{m_2}{d_2}P_2,
\ee
hence the channel is non-unital, unless $m_1=d_1$ and $m_2=d_2$. \\

From the above relations and the fact that  $P_1+P_2=Q_1+Q_2=I$,                                                                                                                                                                                                                                                                                                           we find

\be\label{PP}
P_1=\frac{1}{d}(d_1I+\sqrt{d_1d_2}B)\h P_2=\frac{1}{d}(d_2I-\sqrt{d_1d_2}B)
\ee
and
\be
Q_1=\frac{1}{d}(m_1I+\sqrt{m_1m_2}A)\h Q_2=\frac{1}{d}(m_2I-\sqrt{m_1m_2}A)
\ee
Consider now a state
\be
\rho=\frac{1}{d}(1+xA)
\ee
where
$x$ represents the Bloch vector of this state,
 since $x=\Tr(\rho A).$
We will then find
\be
{\cal E}(\rho)=\frac{1}{d}(I+x'B),
\ee
where
\be\label{x'}
x'=\sqrt{\frac{m_1m_2}{d_1d_2}}x+
\frac{m_1d_2-m_2d_1}{d\sqrt{d_1d_2}} .
\ee

Straightforward calculation completes the reasoning.
Hence we find that if $m_1m_2>d_1d_2$, then the Bloch vector can be stretched.
It is interesting to note that the inhomogeneous translation does not
compensate the stretching of $x$, rather it enhances it. Obviously,  this is not
possible in dimensions $d=2, 3$ but it is possible in dimensions $d\geq 4$,
where for example we can take $(m_1,m_2)=(2,2)$ and $(d_1,d_2)=(1,3)$.
Moreover we see that if the channel is unital, then $x'=x$.\\

To see how the existence of these channels may affect
quasi-inversion, let us denote by $\c M_s$ a subset of quantum channels
whose distortion matrices $M_s$ can only shrink the Bloch vector, i.e.
$\c M_s=\{\c E_s|\ M_s^\T M_s\leq I\}$.
Let us highlight two remarks related to the set $\c M_s$.
First, according to Russo-Dye theorem any linear positive map obtains
its norm at the identity \cite{b07}. Thus, singular values of any unital map
are less than or equal to one which implies all unital maps belong to $\c M_s$.
Moreover, for any $M_s$ included in the set $\c M_s$
it is always possible to find unitary operators $U_1$ and $U_2$
such that $M_s=(U_1+U_2)/2$.
For now, let us assume that quasi-inverse for a quantum channel $\c E$
belongs to $\c M_s$. In that case, one has \eqref{relation of corrected fidelity}
\begin{eqnarray}\label{up-unital}
\overline{F}(\c E^{qi}\circ\c E)&=&\max_{\c E'}
\overline{F}(\c E'\circ\c E)=\max_{\c E_s\in\c M_s}
\overline{F}(\c E_s\circ\c E)\nonumber\\
&=&\max_{M_s} \frac1d\left(1+\frac{1}{d+1}\Tr M_sM\right)
\nonumber\\&\leq&
\max_{U_1,U_2} \frac1d\left(1+\frac{1}{2(d+1)}\Tr[(U_1+U_2)M]\right)
\nonumber\\&=& \frac1d\left(1+\frac{1}{d+1}\Tr|M|\right).
\end{eqnarray}
This inequality imposes an upper bound on the corrected fidelity whenever
quasi-inversion lies in the set $\c M_s$. An example
of which is a qubit channel where not only quasi-inversion is a unital and
unitary map in $\c M_s$, but also the entire set of qubit
channels are contained in $\c M_s$. However, in higher dimensions
$\c M_s$
is a nontrivial subset of quantum channels as our above example shows.
It is a question then whether quasi-inversion of any quantum channel
belongs to the set $\c M_s$. In what follows, we show with
an example it is possible to exceed the bound in \eqref{up-unital},
which implies a negative answer to the question.
Moreover, as the bound \eqref{up-unital} is valid when
quasi-inversion is a unital map, one infers in higher
dimensions, against qubit channel, quasi-inverse is not necessarily unital.\\

Before presenting the example, we mention explicitly taking
$\Gamma_1=\sqrt{d-1}A$ and $\Gamma_2=\sqrt{d-1}B$,
see Eq.~\eqref{A and B}, along
with the identity and $d^2-3$ other Hermitian and traceless operators
orthogonal to $A$ and $B$ as the set of basis,
the distortion matrix $M$ and the translation vector $\vect{t}$
of the channel $\c E$ in the beginning of this section
 are given by their entries as
\begin{equation}\label{e}
M_{ij}=\sqrt{\frac{m_1m_2}{d_1d_2}}\delta_{i2}\delta_{j1},
\quad\quad\quad
t_i=\frac{m_1d_2-m_2d_1}{d\sqrt{d-1}\sqrt{d_1d_2}}\delta_{i2}.
\end{equation}
Now assume $\lambda>0$ is sufficiently small
so the $d$-dimensional map $\c E'$
described by the following affine parameters is a quantum channel
\begin{equation}\label{e'}
M'_{ij}=\lambda\delta_{i1}\delta_{j2},
\quad\quad\quad
t'_i=0.
\end{equation}
Note that $\Tr(MM')=\sqrt{\frac{m_1m_2}{d_1d_2}}\lambda$.
This amount is certainly larger than $\Tr|M'|=\lambda$ if
$m_1m_2>d_1d_2$, i.e. if the affine matrix $M$ can stretch
the Bloch vector and $\c E$ is not contained in $\c M_s$.
So for the channel $\c E'$ specified in Eq. \eqref{e'} we have through
Eq.~\eqref{relation of corrected fidelity}
\begin{equation}
\overline{F}\left(\c E^{\prime qi}\circ\c E'\right)
\geq\frac{1}{d}\left[1+\frac{\lambda}{d+1}\max\sqrt{\frac{m_1m_2}{d_1d_2}}\right].
\end{equation}
For an even $d$ one has $\max\sqrt{\frac{m_1m_2}{d_1d_2}}=\frac{d}{2\sqrt{d-1}}$, while
$\max\sqrt{\frac{m_1m_2}{d_1d_2}}=\frac{\sqrt{d+1}}{2}$
for an odd $d$. In any case, assuming $d>3$ we get
\begin{equation}
\overline{F}\left(\c E^{\prime qi}\circ\c E'\right)
\geq\frac{1}{d}\left[1+\frac{\lambda}{2\sqrt{d+1}}\right]>
\frac{1}{d}\left[1+\frac{\lambda}{d+1}\right]=
\frac{1}{d}\left[1+\frac{\Tr|M'|}{d+1}\right].
\end{equation}

\section{Statistical Properties of Classical Channels and their quasi-inverses}\label{statis-cl}
In this section we discuss some of the properties of classical channels and their inverses.
			A random classical channel is a stochastic matrix $T$ in which the columns are independently picked from the uniform ensemble over the probabilistic simplex
		\be \Delta_d\equiv\big\{\vect{p}\in\mathbb{R}^d\hspace{1mm}\big|\hspace{1mm}p_i\geq0;\hspace{1mm}\sum_ip_i=1\big\}\ee
		Let the Probability Distribution Function (PDF) for any component $p_i$ be denoted by $f_d(x)$, that
		is $\mathbb{P}[x\leq  p_i\leq x+dx]=f_d(x)dx\ \ \ \forall \ i$.
		 Noting as an example figure (\ref{slab}) for the $d=3$ case,
		we see that this PDF is proportional to the area
		 of the narrow slab on the triangle which defines the probability simplex.
			\begin{figure}[H]\vspace{-1cm}
			$${\includegraphics[scale=0.33]{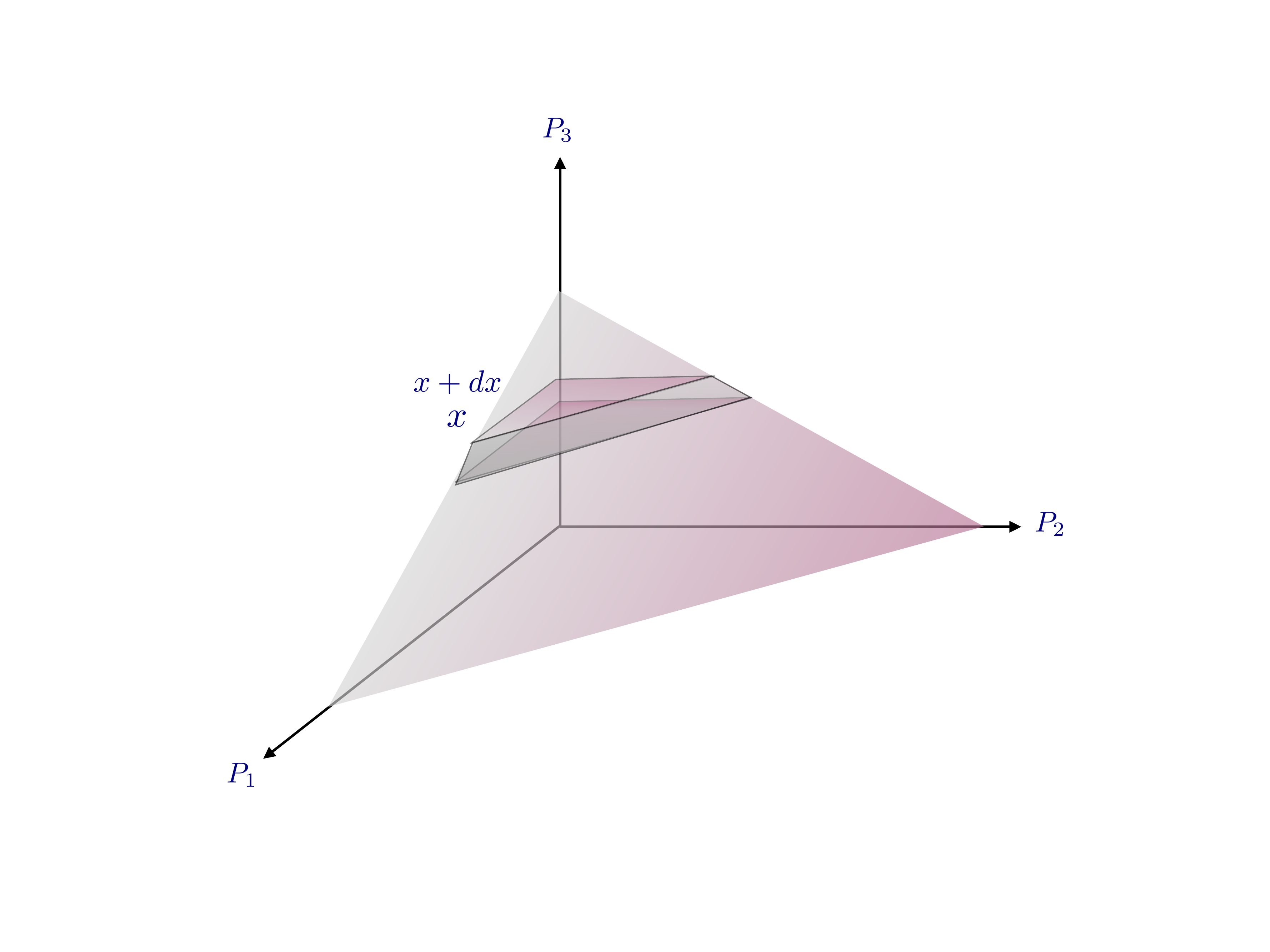}}$$\vspace{-2.2cm}
			\caption{A probability simplex in 3 dimensions. Each point on the large triangle represents a valid probability vector.}
			\label{slab}
		\end{figure}
		For the general case, this function is given by
		\be \phi_d(x)\equiv(d-1)(1-x)^{d-2};\hspace{4mm}x\in[0, 1].\ee
		It is also useful to have an expression for the corresponding CDF
		\be \Phi_d(x)=\mathbb{P}[p_i\leq x]=\int_0^x\phi_d(t)dt=1-(1-x)^{d-1}.\ee
		The average fidelity for any classical channel $T$ is given by $\overline{\c F}(T)\equiv \frac{1}{d}\Tr T$. If we now take the average over the uniform ensemble of all stochastic matrices, we find
		\be \langle \overline{\c F}\rangle=\frac{1}{d}\sum_i\langle T_{ii}\rangle=\int_0^1(d-1)(1-x)^{d-2}xdx=\frac{1}{d},\ee
		where we have used the statistical independence of the columns of $T$ in the uniform ensemble. Moreover we can also find the variance of this average fidelity which turns out to be
		\be \operatorname{Var}\overline{\c F}=\frac{1}{d^2}\sum_i\operatorname{Var}T_{ii}=\frac{d-1}{d^3(d+1)}.\ee
		
		The interesting question is how much this average fidelity increases for classical channels when we apply the quasi-inversion. 		
		To find this we note that the average fidelity after quasi inversion is
	\be \overline{\c F}(T^{qi}\circ T) =\frac{1}{d}\sum_i\max_j T_{ij},\ee
		that is the average fidelity after quasi-inversion is the average of the maximum element in each row of $T$.
		Assuming that the largest element of two different rows do not occur on the same column (which can happen only for a subset of measure zero in the ensemble), the ensemble average of the improved average fidelity becomes
			\be\la \overline{\c F}(T^{qi}\circ T)\ra =\la  p_{m}\ra,\ee
		where $p_{m}$ is the largest element in a single probability vector chosen uniformly.
		To find the average of this quantity for the uniform ensemble, we invoke the proposition that only on a set of measure zero, the maximum values of two different rows may occur on the same column.  Therefore we first find the following Cumulative Distribution Function (CDF), i.e.  the probability that the maximum values in all columns are less than $x$
		$$\Psi_d(x):=\mathbb{P}[p_m\leq x, \hspace{1mm}\forall j]=\Phi_d(x)^d=\big[1-(1-x)^{d-1}\big]^d$$
		this leads to the following PDF, i.e. the probability that the maximum value is between $x$ and $x+dx$,
		$$\psi_d(x)=\frac{d\Psi_d(x)}{dx}=d(d-1)\big[1-(1-x)^{d-1}\big]^{d-1}(1-x)^{d-2}$$
		Now it is easy to see that
		
		\ba\la \overline{\c F}(T^{qi}\circ T)\ra=\langle p_m\rangle&=&1-\langle 1-p_m\rangle=1-d(d-1)\int_0^1dx\big[x(1-x^{d-1})\big]^{d-1}\cr
			&=&1+\sum_{n=0}^{d}\frac{(-1)^n(d-1)n}{(d-1)n+1}\begin{pmatrix}d\\n\end{pmatrix}.
		\ea
		Figure \ref{f*} shows the average fidelity after quasi-inversion versus dimension. The curve fits the equation
		\be
		\la \overline{\c F}(T^{qi}\circ T)\ra=c d^{x}\sim d^{-1/2},
		\ee
		with $c\approx 1.055$ and the exponent $x\approx -0.5042$.
Therefore, the average fidelity of a classical channel
appears to be increasing due to quasi-inversion from $1/d$ to $1/\sqrt{d}$.

			\begin{figure}[H]
					\centering
				\includegraphics[scale=0.45]{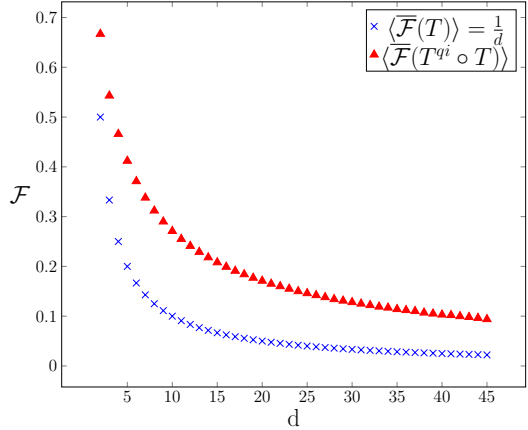}\vspace*{-5cm}
						\caption{The average fidelity of a classical channel
	before $\langle\overline{\c F}(T)\rangle$ (blue cross) and
	after $\langle\overline{\c F}(T^{qi}\circ T)\rangle$ (red triangle)
	it is compensated by the quasi-inversion.
						}
			\label{f*}
		\end{figure}
\section{ A note on statistical Properties of Quantum Channels and their quasi-inverses}\label{statis-qu}

Investigating the same problem in the set of quantum channels seems
a highly difficult task since in quantum case the quasi-inversion is not
generally known. Nonetheless, numerical analysis is possible in lower
dimensions and it shows improvement in amount of fidelity
for a typical channel.
To see that, one may notice for random quantum channels distributed
uniformly in the set of quantum channels \cite{karol3}
\begin{equation}
\langle\Phi_{\c E}\rangle=\Phi_*,
\end{equation}
where $\Phi_*$ is the maximally depolarizing channel,
see Section \ref{prem}. This implies $\langle\Tr\Phi_{\c E}\rangle=1$.
Thus, the average amount of input-output fidelity over random channels
distributed uniformly in $d$ dimension is given by \eqref{fbar}
\begin{equation}
\langle\overline{F}(\c E)\rangle=\frac1d.
\end{equation}
To get the fidelity after correction averaged over the set of uniformly
distributed quantum channels, we applied  numerically searching
for the quasi-inverse. The sketch of our method is based on
the observation that any quantum channel ${\cal E}:D(H_d)\lo D(H_d)$
should satisfy a continuous family of inequalities
\be
\la \phi|({\cal E}\otimes I)(|\psi\ra\la \psi|)|\phi\ra\geq 0
\h |\phi\ra, \ |\psi\ra\ \in H_d\otimes H_d.
\ee
To approximate the quasi-inverse of a channel numerically,
we first approximate the set of CP operators by only including a
finite subset of such constraints obtained by random entangled states in
$H_d\otimes H_d$. Then the quasi-inverse will be calculated as the answer to
a finite linear programming problem using the simplex method or other efficient algorithms.
The results for dimensions $d=2,3,4,5$ are presented in
Fig.~\ref{CorrectedF} and they confirm improvement in the average fidelity
after correction for random channels.
\begin{figure}[!h]
\centering
\includegraphics[scale=0.45]{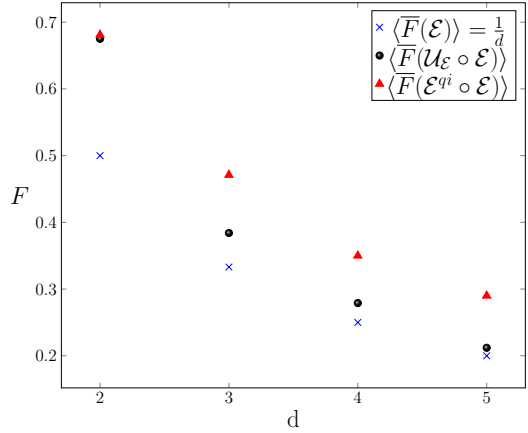}\vspace*{-5cm}
\caption{The amount of input-output fidelity
$\langle\overline{F}(\c E)\rangle$ (blue cross),
fidelity after correction with the best possible unitary
$\langle\overline{F}(\c U_{\c E}\circ\c E)\rangle$ (black circle),
and fidelity after correction with quasi-inversion
$\langle\overline{F}(\c E^{qi}\circ\c E)\rangle$ (red triangle)
averaged over the set of random channels distributed uniformly
for $d=2,\dots,5$.}
\label{CorrectedF}
\end{figure}

One may also think about applying a unitary evolution to correct the
input-output fidelity. However, our numerical results show that the best
unitary quantum channel, which will be denoted by $\c U_\c E$,
 cannot modify the fidelity in the best possible way,
 see Fig.~\ref{CorrectedF} for a comparison.
Indeed, we can take the purity of Jamio{\l}kowski state as a signature
of unitarity of a given channel $\c E$, i.e. this quantity is equal to $1$
if and only if the channel is unitary and any deviation from $1$ shows
non-unitarity. Adopting such a function, one finds
unitarity for the quasi-inverse of random channels averaged over
the set of uniformly distributed CPT maps is equal to
$0.99\pm0.01$ for $d=3$, $0.68\pm0.05$ for $d=4$, and
$0.52\pm0.05$ for $d=5$. This interesting fact counter-intuitively
shows in higher dimensions for a typical channel quasi inversion
is almost a non-unitary channel. However, if we measure unitality by
$1-|\vect{t}|^2$ ($\vect{t}$ is the translation vector, see \eqref{affine})
averaged over the set of random channels,
we see for $d=3,4,5$ it is respectively given by:
$0.98\pm0.01$, $0.97\pm0.01$, and $0.97\pm0.01$, which confirms
that the quasi-inverse of a typical channel is close to be unital.


\end{document}